\documentclass[11pt]{amsart}
\usepackage{amsmath,amsfonts,amscd,amssymb,amsthm,epsfig,euscript}
\usepackage{mathrsfs}
\usepackage{amsxtra}
\usepackage[all]{xy}
\usepackage{verbatim,epsf}
\usepackage{framed}
\newtheorem{theorem}{Theorem}
\newtheorem{proposition}{Proposition}
\newtheorem{lemma}{Lemma}
\newtheorem{corollary}{Corollary}

\theoremstyle{definition}
\newtheorem{definition}{Definition}

\theoremstyle{remark}
\newtheorem{remark}{Remark}
\numberwithin{equation}{section}
\newcommand{\field}[1]{\ensuremath{\mathbb{#1}}}

\newcommand{\DD}{\field{D}}

\newcommand{\RR}{\field{R}}

\newcommand{\ZZ}{\field{Z}}

\newcommand{\R}{{\mathbb{R}}}
\newcommand{\curly}[1]{\mathscr{#1}}
\newcommand{\cA}{\curly{A}}
\newcommand{\cB}{\curly{B}}
\newcommand{\cC}{\curly{C}}
\newcommand{\cD}{\curly{D}}
\newcommand{\cE}{\curly{E}}
\newcommand{\cF}{\curly{F}}
\newcommand{\cG}{\curly{G}}
\newcommand{\cH}{\curly{H}}
\newcommand{\cI}{\curly{I}}
\newcommand{\cJ}{\curly{J}}
\newcommand{\cK}{\curly{K}}
\newcommand{\cL}{\curly{L}}
\newcommand{\cM}{\curly{M}}
\newcommand{\cN}{\curly{N}}

\newcommand{\cP}{\curly{P}}

\newcommand{\cU}{\curly{U}}
\newcommand{\cV}{\curly{V}}
\newcommand{\cW}{\curly{W}}

\newcommand{\cZ}{\curly{Z}}

\newcommand{\Ga}{\Gamma}

 \DeclareMathOperator{\Hom}{Hom}
 
\DeclareMathOperator{\Aut}{Aut}

\usepackage{hyperref}
\hypersetup{colorlinks,citecolor=blue,plainpages=false,hypertexnames=false}
\begin{document}

\title[Homotopy classes of gauge fields and the lattice]{Homotopy classes of gauge fields and the lattice}


\author[Meneses]{Claudio Meneses}
\address{\noindent Centro de Investigaci\'on en Matem\'aticas A.C., Jalisco S/N,
 Valenciana, \indent C.P. 36023, Guanajuato, Mexico}
\curraddr{Mathematisches Seminar, Christian-Albrechts Universit\"at zu Kiel, \indent Ludewig-Meyn-Str. 4, 
24118 Kiel, Germany}
\email{meneses@math.uni-kiel.de} 

\author[Zapata]{Jos\'e A. Zapata}
\address{Centro de Ciencias Matem\'aticas, Universidad Nacional Aut\'onoma de M\'exico
\indent UNAM-Campus Morelia, A. Postal 61-3, Morelia, Mexico, and Department 
\indent of Applied Mathematics, University of Waterloo, Waterloo, Canada}
\email{zapata@matmor.unam.mx} 

\thanks{}

\thanks{}

\subjclass[2010]{Primary 	55P35, 55R10, 70S15, 81T25}

\date{}

\dedicatory{}
\begin{abstract}
For a smooth manifold $M$, possibly with boundary and corners, and a Lie group $G$, we consider a suitable description of gauge fields in terms of parallel transport, as groupoid homomorphisms from a certain path groupoid in $M$ to $G$.  Using a cotriangulation $\cC$ of $M$, and collections of finite-dimensional families of paths relative to $\cC$, we define a homotopical equivalence relation of parallel transport maps, leading to the concept of an extended lattice gauge (ELG) field. A lattice gauge field, as used in Lattice Gauge Theory, is part of the data contained in an ELG field, but the latter contains further local topological information sufficient to reconstruct a principal $G$-bundle on $M$ up to equivalence. The space of ELG fields of a given pair $(M,\cC)$ is a covering for the space of fields in Lattice Gauge Theory, whose connected components parametrize equivalence classes of principal $G$-bundles on $M$. We give a criterion to determine when ELG fields over different cotriangulations define equivalent bundles.
\end{abstract}

\maketitle


\tableofcontents
\section{Introduction}

The standard geometric approach to the study of field theories with gauge symmetry such as Yang--Mills theory begins with a principal $G$-bundle $\pi: P\to M$ over a smooth manifold $M$. From it, a space of states $\cA_{P}/\cG_{P}$ is induced---the space of gauge fields---consisting of smooth connections in $P$ modulo the action of the group of gauge transformations. In the usual formulation of Lattice Gauge Theory (LGT) 
\cite{Creu83}, the discrete analogs of a gauge field, the so-called \emph{lattice gauge fields}, arise as evaluations of the parallel transport of a connection over a discrete collection of paths gene\-rated by a lattice $\Lambda$. However, the role of the topology of $P$ (an intrinsically global problem) is less evident in the subsequent axiomatization of lattice gauge fields when $M$ is arbitrary. 

Different proposals for a mechanism to consistently associate a bundle to a lattice gauge field exist in the literature since the introduction of LGT. L\"uscher \cite{Luscher82} observed that if the parallel transporters $\{g_{\gamma}\in\mathrm{SU}(2)\}$ on a 4-torus satisfy certain metric constraints, then an interpolation algorithm could be defined on such a LG field to construct the topological charge of a principal bundle. In a series of works \cite{Phil85,PS86,PS90,PS93}, Phillips and Stone provided a rigorous construction of L\"uscher's algorithm for arbitrary compact $M$ and $G$, in terms of a combinatorial realization of the homotopy of a principal bundle and its characteristic classes. 

The objective of this work is to propose an alternative to the previous ideas, which is somewhat similar in spirit to the combinatorial constructions of Phillips--Stone, but with a crucial difference: instead of truncating the space of LG fields, we complement LG fields, in such a way that the topology of a principal $G$-bundle is incorporated as part of the data in a local manner. Our proposal is motivated by the following proposition, proved in section \ref{sec:cellular data}. Let $M$ be equipped with a network of paths $\Gamma = \{\gamma\}$, and let $\{p_{\sigma}\}$ be its set of vertices. For any given $P$ and a choice of fiber base points $\{b_{\sigma}\in \pi^{-1}(p_{\sigma})\}$, consider the map on $\cA_{P}$ that assigns to any given connection $A\in \cA_{P}$ its set of parallel transporters $\{g_{\gamma}(A)\in G\}$. 

\begin{proposition}\label{prop:interpolation}
For any choice of LG field $\{g_{\gamma}\}$ in $(M,\Gamma)$, and any principal $G$-bundle with a choice of base points $(P,\{b_{\sigma}\})$, there exists a smooth connection $A\in\cA_{P}$ (in fact, infinitely many) such that $g_{\gamma}(A) = g_{\gamma}$.
\end{proposition}

The classification of principal $G$-bundles on a smooth $n$-manifold $M$ can be encoded in the homotopy theory of maps to the classifying space $BG$ \cite{Steen51,Milnor56,NR61,MS74}, the conjugacy classes of covering homotopy homomorphisms 
\cite{Lash56}, and the \v{C}ech cohomology spaces $\check{H}^{1}\left(M,G\right)$. In our proposal, a crucial step consists on interpreting the \v{C}ech cohomology classification combinatorially, in terms of an auxiliary ``scaffolding" structure in $M$, namely a cell decomposition $\cC$ dual to a triangulation in $M$.\footnote{This idea is not new, and goes back to Segal  \cite{Segal68}. A variation of it appears in \cite{PS90,PS93} as the notion of \emph{parallel transport functions} (p.t.f.).} 
The crucial fact that we exploit systematically is the correspondence---via \emph{clutching maps}---between $\check{H}^{1}\left(S^{k},G\right)$, $k\geq 2$ and the homotopy group $\pi_{k-1}(G,e)$, reducing the study of the space $\check{H}^{1}\left(M,G\right)$ to the groups $\pi_{k}(G,e)$, $0 < k < n$.  

Our proposal, leading to the notion of an \emph{extended lattice gauge field} (definition \ref{def:ELGT data}) consists of giving a discretization mechanism for gauge fields 
that preserves the topology of a principal $G$-bundle through the decimation process. We use an arbitrary but fixed LG field as a central ingredient to define relative homotopy classes of compatible transition maps from a bundle trivialization over a $k$-cell to trivializations over its boundary subcells (which we call \emph{glueing maps}), $1 \leq k \leq n$. As a result, not only every ELG field encodes a LG field, but also an isomorphism class of principal $G$-bundles \emph{in a local manner}, getting rid of the indeterminacy in proposition \ref{prop:interpolation}.\footnote{In some cases of physical interest where the topology of a bundle is necessarily trivial (e.g. the dimensions $n = 2,3$, and $G= \mathrm{SU}(r)$; see remark \ref{remark:simple-simply-conn}) the added data is trivial, and our proposed extension yields standard LG fields.} 
The idea behind the definition of an ELG field is based on the Barrett--Kobayashi construction \cite{Koba54,Ba91,Lew93,CP94}, which asserts that an axiomatization of the notion of \emph{holonomy} of a connection \cite{KN63}, based on the structure of a suitable group of equivalence classes of based loops in $M$, is sufficient to recover a pair $(P, [A])$ consisting of a principal $G$-bundle and a gauge orbit of smooth connections on it. 

The framework that we propose has yet another field theoretical motivation. When dealing with a manifold with boundary representing a portion of spacetime, Hamilton's principle of extremal action for classical fields requires the consideration of spaces of fields that are compatible with a given field at the boundary. Similarly, quantum field theory frameworks require the consideration of spaces of fields whose restriction to the boundary is given. 
Our proposal originates from the consideration of such scenarios in the case of gauge fields described through their parallel transport. Consequently, the constructions in our work are suited to study such physically fundamental problems for gauge field theories on base spaces which may have boundary and corners, which are essential in local covariant quantization frameworks such as spin foam models \cite{Reisenberger:1994aw,Baez:1997zt}.\footnote{A precursor of our work in the case of two dimensional abelian gauge field theory is presented in \cite{ORS05}, where it is shown that extra  bundle data complementing ordinary LGT is essential in the construction of a spin foam model for two-dimensional gravity. This is the case because in two-dimensional gravity, there are homotopically non trivial maps from the boundary of a disk to the structure group, which in Euclidean signature is $\mathrm{SO}(2)$. In dimension four the relevant structure group is $\mathrm{SU}(2)$, which again makes the bundle structure data essential.}

\subsection{A structural guide to our work}

The article is organized as follows. We begin by providing a brief glossary of facts an terminology that our work is based on, presented in section \ref{sec:terminology}. The aim is to give a precise and concise introduction to the structural foundations of our constructions: intimacy of paths, groups of based loops, holonomy maps, the Barrett-Kobayashi reconstruction theorem, and lattice gauge fields. To keep the exposition as brief as possible, the additional facts on cotriangulations that we will need are provided in appendix \ref{sec:triangle-dual}. 

Section \ref{sec:path-structures} introduces the main technical notions in our work, namely path groupoids $\cP_{\cC}$ and parallel transport maps relative to a cotriangulation $\cC$, which are used to define ELG fields. In order to do so, the Barrett-Kobayashi reconstruction theorem for bundles with connection is first reformulated accordingly in theorem \ref{theo:PT}, with a proof presented in appendix \ref{sec:parallel transport}. Supplementing $M$ with a network of paths---a discrete subgroupoid of $\cP_{\cC}$---leads to the notion of the LG field of a parallel transport map. Similarly, a special type of \emph{finite-dimensional} local path subgroupoids of $\cP_{\cC}$ are the crucial ingredient to define ELG fields as suitable classes in relative homotopy. 

The definition of ELG fields is chosen for being succinct. However, we study the relation between ELG fields and isomorphism classes of principal $G$-bundles by dissecting the former into local building blocks. The dissected data is described in theorem \ref{theo:dissection-data}, section \ref{sec:cellular data}. The local building blocks contain the usual LGT data, and additionally, homotopy classes of extensions of glueing maps from the boundary of a $(k+1)$-cell---a $k$-sphere---to its interior. This allows us to identify the minimal local topological data contained in an ELG field that is sufficient to reconstruct a bundle, which we name the \emph{core} of an ELG field (corollary \ref{cor:core}).\footnote{The core of an ELG field corresponds to a homotopy class of p.t.f. in the terminology of Phillips--Stone \cite{PS90,PS93}, relative to a given LG field.} 
Allowing the underlying LG field to change continuously in a family of cores of ELG fields leads the notion of \emph{cellular bundle data}, introduced in definition \ref{def:bundle data}. The correspondence between cellular bundle data and isomorphism classes of principal $G$-bundles is proved in theorem \ref{theo:equivalence} in appendix \ref{app:cellular bundle data}. Section \ref{sec:small-dim} is devoted to describing explicitly the dissected data in small dimensions. 

A study of the spaces of ELG fields is given in section \ref{sec:category}. The dissected local data reveals a \emph{covering space} structure in the space of ELG fields for a given triple $(M,\cC,\Gamma)$, fibering over the space of standard LGT data---a Lie group---(corollary \ref{cor:torsor}). The group of deck transformations of the cover is modeled on subgroups of a product of homotopy groups, acting on the homotopy classes of extensions of glueing maps. The connected components of the cover are in bijective correspondence with $\check{H}^{1}\left(M,G\right)$. As a by-product, we reconstruct the space $\check{H}^{1}\left(M,G\right)$ as a \emph{homogeneous space} for such a group of deck transformations (corollary \ref{cor:fibration}). Section \ref{sec:Pachner} describes the dependence of the bundle structure of ELG fields on auxiliary cotriangulations, addressed in terms of Pachner moves. Following the theorem of U. Pachner \cite{Pach91}, we provide an algorithm to determine when two ELG fields on different cell decompositions are homotopic and define equivalent bundles (theorem \ref{theo:Pachner}).

\section{Terminology and fundamental notions we build our work upon}\label{sec:terminology}

Let $M$ be a smooth manifold (possibly with boundary and corners). A path $\gamma:[0,1]\to M$ will always be assumed to be continuous and piecewise-smooth, unless otherwise stated. We follow the convention of denoting $\gamma(0)=s(\gamma)$ and $\gamma(1)=t(\gamma)$ (the \emph{source} and \emph{target} of $\gamma$). In particular, a based loop is a path satisfying $s(\gamma) = t(\gamma)$. We will denote by $\gamma^{-1}$ the inverse path $\gamma^{-1}(t) := \gamma(1-t)$ and the path composition of a pair $(\gamma_{1}, \gamma_{2})$ such that $s(\gamma_{1}) = t(\gamma_{2})$ by $\gamma_{1}\cdot\gamma_{2}$. 
For a choice of base point $p\in M$, $\Pi(M,p)$ will denote the path space of $M$ with $s(\gamma) = p$. There is a natural fibration $\mathrm{pr}:\Pi(M,p) \rightarrow M$ given by $\gamma \mapsto t(\gamma)$. We will denote the fibers $\mathrm{pr}^{-1}(q)$ by $\Pi(M,p,q)$. In particular, $\Pi(M,p,p)$ is the loop space $\Omega(M,p)$. The inversion and composition of loops define operations on $\Omega(M,p)$. 

\begin{definition}[\cite{Ba91,CP94}]\label{def:thin-homotopy}
Two based paths $\gamma_{0}, \gamma_{1}\in \Pi(M,p,q)$ are \emph{thin homotopic} if there is a piecewise-smooth homotopy $\gamma:[0,1]^{2}\rightarrow M$ such that $\gamma(s,\cdot)\in  \Pi(M,p,q)$ $\forall s\in[0,1]$, $\gamma(i,\cdot) = \gamma_{i}$ for $i = 0,1$, and $\gamma\left([0,1]^{2}\right)\subseteq \gamma_{0}([0,1])\cup \gamma_{1}([0,1])$. 
\end{definition}

\begin{remark}
If $p = q$ and $\gamma_{1},\gamma_{2}$ are based loops, definition \ref{def:thin-homotopy} is a reformulation of Barrett's  \cite{Ba91,CP94}. Let $\sim_{R_{1}}$ be the equivalence relation in $\Pi(M,p,q)$ given by decreeing $\gamma_{0}\sim_{R_{1}}\gamma_{1}$ if there is a finite sequence of thin homotopies connecting them.  A \emph{path retracing} \cite{AL94,GP96} is any thin homotopy between a path of the form $\gamma_{1}\cdot\gamma^{-1}\cdot\gamma\cdot\gamma_{2}$ and its reduction $\gamma_{1}\cdot\gamma_{2}$. It is
a fundamental question, postulated in \cite{Ba91,AL94,GP96}, whether two paths $\gamma_{0}, \gamma_{1}\in \Pi(M,p,q)$ are thin homotopic if and only if they differ by a finite collection of piecewise-smooth reparametrizations and retracings. Such correspondence would provide a useful characterization (and a possible axiomatization) of thin homotopy  in terms of ``generators and relations".  
\end{remark}

Let $P(M,p)\subset \Pi(M,p)$ (resp. $P(M,p,q)\subset \Pi(M,p,q)$) the subspace of smooth paths with \emph{sitting instants} at $0,1$ \cite{CP94} i.e. smooth paths that are constant in a neigborhood of $0$ and $1$. Similarly, let $L(M,p) = P(M,p,p)$. In general, there exist a retract $\Pi(M,p,q)\rightarrow P(M,p,q)$, since every  path $\gamma\in \Pi(M,p,q)$ can always be reparametrized to a new path $\gamma' \in P(M,p,q)$.

\begin{definition}[\cite{CP94}, cf. \cite{SW09,BH11}]\label{def:intimacy}
Two paths $\gamma_{0}, \gamma_{1}\in P(M,p,q)$ are \emph{intimate} if there is a smooth map $\gamma:[0,1]^{2}\rightarrow M$ and $0<\epsilon <1/2$ such that $\gamma(s,\cdot)\in P(M,p,q)$ $\forall s\in [0,1]$, $\gamma(s,\cdot) = \gamma_{0}$ and $\gamma(1-s,\cdot)=\gamma_{1}$ for $0 \leq s\leq\epsilon$, and  $\mathrm{Rank}\left( d \gamma\right) < 2$ everywhere.
\end{definition}

\begin{remark}
There is an analogous equivalence relation $\sim_{R_{2}}$ in the spaces $P(M,p,q)$ given by decreeing $\gamma_{0}\sim_{R_{2}}\gamma_{1}$ if they are intimate.
Both thin homotopy and intimacy possess a fundamental feature. The operations of (possibly reparametrized) inversion and composition of based loops descend to equivalence classes as $[\gamma]^{-1} := \left[\gamma^{-1}\right]$ and $[\gamma_{1}]\cdot[\gamma_{2}]:= [\gamma_{1}\cdot\gamma_{2}]$, leading to an induced \emph{topological group} structure on the quotient spaces $\Omega(M,p)/\sim_{R_{1}}$ and $L(M,p)/\sim_{R_{2}}$. 
It is shown in \cite{CP94} that the relation $\sim_{R_{2}}$ is an explicit example of a \emph{holonomy relation}, i.e., if a pair  $\gamma_{0},\gamma_{1}\in L(M,p)$ satisfies $\gamma_{0}\sim_{R_{2}}\gamma_{1}$, then $\mathrm{Hol}(A,\gamma_{0}) = \mathrm{Hol}(A,\gamma_{1})$ for any smooth connection $A$ on any principal $G$-bundle $P\rightarrow M$. The work \cite{CP94} was motivated by the following subtle observation: it is not straightforward to show that thin homotopy is a holonomy relation.

\end{remark}

We will denote the quotient $P(M,p,q)/\sim_{R_{2}}$ by $\cP(M,p,q)$. We can analogously  define $\cP(M,p)$ through the corresponding fibration over $M$ with fibers $\cP(M,p,q)$. We will call $\cL(M,p) := L(M,p)/\sim_{R_{2}}$ the \emph{group of based loops} of $M$.
Given a pair $(P,b)$ of a principal $G$-bundle $\pi: P\rightarrow M$ and a choice $b\in \pi^{-1}(p)$, holonomy based at $b$ defines a map
\[
A\in\cA_{P} \mapsto \mathrm{Hol}(A)\in \Hom(\cL(M,p),G)_{0}
\]
where $\Hom(\cL(M,p),G)_{0}$ is the set of topological group homomorphisms $\cL(M,p)\rightarrow G$ satisfying a suitable smoothness condition (condition H3 in \cite{Ba91}; c.f. definition \ref{def:PT}). 
The holonomy map is invariant under the action of the group $\cG_{P,p}$ of gauge transformations acting trivially on $\pi^{-1}(p)$ (or equivalently, preserving $b$), and descends to a map $\cA_{P}/\cG_{P,p}\rightarrow \Hom(\cL(M,p),G)_{0}$. The Barrett-Kobayashi reconstruction theorem can be stated as follows.

\begin{theorem}[\cite{Ba91,Lew93,CP94}]
Given a choice of based principal $G$-bundles $\left(P,b\in\pi^{-1}(p)\right)$ for every isomorphism class $\{P\}\in \check{H}^{1}(M,G)$, the induced map 
\begin{equation}\label{eq:B-K}
\mathrm{Hol}:\bigsqcup_{\{P\}\in \check{H}^{1}(M,G)}\cA_{P}/\cG_{P,p}\longrightarrow \Hom(\cL(M,p),G)_{0} 
\end{equation}
is a bijection.
\end{theorem}

\begin{remark}
The natural function space topologies of $\Hom(\cL(M,p),G)_{0}$ are not explicitly discussed in the works \cite{Ba91,Lew93,AL94,CP94,GP96}. In this work we will only consider its topology induced by the identification \eqref{eq:B-K}. In particular, there is an induced bijection (cf. \cite{Lash56})
\begin{equation}
\pi_{0}(\Hom(\cL(M,p),G)_{0}) \longleftrightarrow \check{H}^{1}(M,G).
\end{equation}
\end{remark}

Let $\cC = \{c_{\tau}\}$ be a cotriangulation on $M$ with prescribed base points $\cB_{\cC} = \{p_{\tau}\in c_{\tau}\,|\,c_{\tau}\in\cC\}$, and for every $c_{\tau}\in \cC$ let $\cB_{\tau} = \{p_{\sigma}\in c_{\sigma}\,|\,\overline{c_{\sigma}}\subset\overline{c_{\tau}}\}$. Details on cotriangulations are presented in appendix \ref{sec:triangle-dual}. 

\begin{definition}\label{def:cellular network}
 A  \emph{cellular network} $\cP_{\Gamma}$ in a cotriangulation $\cC$ of $M$ with base points $\cB_{\cC}$ is the groupoid generated by a collection of smooth paths $\Gamma =\{\gamma_{\sigma\tau}\in P(\overline{c_{\tau}},p_{\tau},p_{\sigma})\}$ indexed by pairs of an arbitrary $c_{\tau}\in\cC$ and any 0-cell $c_{\sigma} = p_{\sigma}\in\partial\overline{c_{\tau}}$, with images in the 1-skeleton $\mathrm{Sk}_{1}(B(\cC))$ of a barycentric subdivision of $\cC$. A \emph{lattice gauge (LG) field} on $(M,\Gamma)$ is a groupoid homomorphism $g:\cP_{\Gamma}\rightarrow G$.
 \end{definition}

\section{The definition of an extended lattice gauge field}\label{sec:path-structures}

\begin{definition}\label{def:path groupoid}
Let 
\[
\cP_{\cC} := \bigcup_{p_{\sigma_{1}},p_{\sigma_{2}}\in\cB_{\cC}}\cP(M,p_{\sigma_{1}},p_{\sigma_{2}}), 
\]
and for every $c_{\tau}\in\cC$, let 
\[
\cP_{\tau} := \bigcup_{p_{\sigma_{1}},p_{\sigma_{2}}\in\cB_{\tau}}\cP(\overline{c_{\tau}},p_{\sigma_{1}},p_{\sigma_{2}}).
\]
\end{definition}

\begin{remark}
There are obvious inclusions $\iota_{\tau}:\cP{\tau}\hookrightarrow \cP_{\cC}$ $\forall c_{\tau}\in\cC$. By a slight abuse of notation, we will identify $\cP_{\tau}$ with its image in $\cP_{\cC}$ whenever the context requires it.
\end{remark}

\begin{remark}
The definition of intimacy ensures that $\cP_{\cC}$ and each $\cP_{\tau}$ are groupoids under the operation of path composition. In particular, the groupoids $\cP_{\tau}$ contains the loop groups $\cL(\overline{c_{\tau}}, p_{\sigma})$ $\forall p_{\sigma}\in\cB_{\tau}$. 
Any cellular network relative to $\cC$ is a discrete subgroupoid  $\cP_{\Gamma}\subset \cP_{\cC}$.
\end{remark}
 
 \begin{definition}
Let $\DD^{r}$ denote the unit disk in $\RR^{r}$. For any pair $c_{\sigma_{1}}, c_{\sigma_{2}}\in\cC$, an $r$\emph{-dimensional smooth path family} $\cF^{r}_{\sigma_{1} \sigma_{2}}$ is a map $\phi:\DD^{r}\rightarrow \cP(M,p_{\sigma_{1}},p_{\sigma_{2}})$ induced by any smooth map
\[
f:\DD^{r}\times[0,1]\rightarrow M
\]
such that $\gamma^{x} := f(x,\cdot)\in P(M,p_{\sigma_{1}},p_{\sigma_{2}})$ $\forall x\in \cU$.
\end{definition}
 
\begin{definition}\label{def:PT}
Let $G$ be a Lie group. A \emph{smooth parallel transport map}, relative to a choice of cotriangulation $\cC$ of $M$, is a groupoid homomorphism
\[
\mathrm{PT}_{\cC}:\cP_{\cC}\to G,
\]
 such that for any choice of local path family $\cF_{\sigma_{1}\sigma_{2}}$, the map induced by evaluation
\[
g_{\sigma_{1}\sigma_{2}}: \DD^{r}\rightarrow G 
\] 
is smooth. Two parallel transport maps $\mathrm{PT}_{\cC}$ and $\mathrm{PT}'_{\cC}$ are equivalent if there is a set $\{g_{\tau}\in G\,|\, c_{\tau}\in\cC\}$ such that for any $[\gamma]\in\cP(M,p_{\sigma_{1}},p_{\sigma_{2}})$
\[
\mathrm{PT}'_{\cC}([\gamma])=g_{\sigma_{2}}\mathrm{PT}_{\cC}([\gamma])g_{\sigma_{1}}^{-1}.
\]
\end{definition}


\begin{remark}
Let $G_{M}$ be the groupoid (thought of as a covariant functor) which to every $p\in M$ assigns a $G$-principal homogeneous space $P_{p}$, and to every morphism $p\to q$, $p,q\in M$, assigns a $G$-equivariant map $P_{p}\to P_{q}$. More generally, a smooth parallel transport map could be defined as a groupoid homomorphism $\mathrm{PT}: \cP_{M} \to G_{M}$ satisfying an analogous smoothness condition, where $\cP_{M} = \bigcup_{p,q\in M}\cP(M,p,q)$ (cf. \cite{SW09}).  
\end{remark}

\begin{remark}
The motivation for definition \ref{def:PT} is the following. Let $\cG_{P,*}$ be the group of gauge transformations of a principal $G$-bundle $\pi:P\to M$ acting trivially along the fibers $\left\{\pi^{-1}\left(p_{\sigma}\right)\right\}_{c_{\sigma}\in\cC}$. For any choice of base points $\cE_{\cC} := \{b_{\sigma} \in\pi^{-1}\left( p_{\sigma}\right)\,|\,c_{\sigma}\in\cC\}$, any isomorphism class of triples
\[
\left[\left(P,\cE_{\cC}, [A]\in\cA_{P}/\cG_{P,*}\right)\right]
\]
determines an equivalence class of smooth parallel transport maps $\left\{\mathrm{PT}_{\cC}\right\}$ via horizontal lifts and parallel transport. That the previous correspondence is bijective is just a reformulation of the theorem of Barrett-Kobayashi in terms of parallel transport. This fact is established in the following result. Its proof is given in appendix \ref{sec:parallel transport}.
\end{remark}

\begin{theorem}[Reconstruction theorem]\label{theo:PT}
Let $M$ be a smooth $n$-manifold with a cotriangulation $\cC$, and $G$ a Lie group. There is a bijective correspondence 
\[
 \left\{\parbox[d]{2.4in}{\centering $G$-valued smooth parallel transport maps $\mathrm{PT}_{\cC}$ in $M$,  up to equivalence}\right\} \leftrightarrow \left\{\parbox[d]{1.4in}{\centering Classes $\left[\left(P,\cE_{\cC}, [A]\right)\right]$, $[A]\in\cA_{P}/\cG_{P,*}$
 }\right\}.
\]
\end{theorem}

\begin{remark}
When $\partial M \neq \emptyset$ and $\partial \cC := \cC|_{\partial M}$ is also a cotriangulation of $\partial M$, let $\partial P := P|_{\partial M}$. We can consider a fixed gauge field  $[a]\in \cA_{\partial P}/\cG_{\partial P,*}$ with corresponding parallel transport map $\mathrm{pt}_{\partial \cC}:\cP_{\partial \cC}\rightarrow G$ where $\cP_{\partial \cC}$ is the corresponding path groupoid on $\partial M$. As a corollary of theorem \ref{theo:PT}, restriction of gauge field data to the boundary of $M$ induces an additional correspondence 
\[
 \left\{\parbox[d]{2.4in}{\centering \emph{$G$-valued smooth parallel transport maps $\mathrm{PT}_{\cC}$ in $M$ such that $\mathrm{PT}_{\cC}|_{\cP_{\partial \cC}} = \mathrm{pt}_{\partial \cC}$, up to equivalence}}\right\} \leftrightarrow \left\{\parbox[d]{1.6in}{\centering \emph{Classes} $\left[\left(P,\cE_{\cC}, [A]\right)\right]$, $[A]\in\cA_{P}/\cG_{P,*}$ such that $[A]|_{\partial P} = [a]$
 }\right\}.
\]
\end{remark}

It will be very convenient to introduce a special class of path families $\cF_{\sigma\tau}$ compatible with a cotriangulation $\cC$, which will be called  \emph{cellular path families}, and which will be considered exclusively henceforth.

\begin{definition}
Given a set of based diffeomorphisms 
\[
\cD = \left\{\psi_{\tau}:\overline{\mathbb{D}^{k}}\rightarrow\overline{c_{\tau}},\, \psi_{\tau}(0) = p_{\tau} \,|\, c_{\tau}\in\cC_{k},\, k =1,\dots, n\right\}
\]
its induced \emph{cellular path groupoid} is the topological subgroupoid $\cP_{\cD}\subset \cP_{\cC}$ generated by the finite-dimensional families $\cF_{\sigma\tau}$ of paths 
\[
\gamma^{x}_{\sigma\tau} := \left(\gamma_{\sigma}^{x}\right)^{-1}\cdot\gamma_{\tau}^{x} \in\cP(\overline{c_{\tau}},p_{\tau},p_{\sigma}),\qquad x\in\overline{c_{\sigma}}\subset\overline{c_{\tau}}
\]
where $\gamma_{\sigma}^{x}$ and $\gamma_{\tau}^{x}$ are the images under $\psi_{\sigma}$ and $\psi_{\tau}$ of the corresponding linear rays starting at 0 and such that $t\left(\gamma_{\sigma}^{x}\right)=t\left(\gamma_{\tau}^{x}\right)=x$.\footnote{A continuous and piecewise-smooth structure on the family $\cF_{\sigma\tau}$ is given in terms of its bijection with $\overline{c_{\sigma}}$. In particular, $\dim\cF_{\sigma\tau} = \dim c_{\sigma}$.} In the special case when $x = p_{\sigma}$, we will denote $\gamma^{x}_{\sigma \tau}$ simply by $\gamma_{\sigma\tau}$.
\end{definition}

%

Any cellular path groupoid $\cP_{\cD}\subset\cP_{\cC}$ contains in particular a cellular network groupoid $\cP_{\Gamma}\subset\cP_{\cD}$ on it, generated by the paths $\gamma_{\sigma\tau}^{p_{\sigma'}}$ for any $c_{\sigma'}\subset \overline{c_{\sigma}}$. Such $\cP_{\Gamma}$ then determines a LG field for any parallel transport map $\mathrm{PT}_{\cC}$ by restriction. The importance of cellular path groupoids relies on the fact that a choice of them can be used to define a homotopy equivalence relation on parallel transport maps relative to its cellular network. 

\begin{definition}[(cf. \cite{Lash56}, definition 3.5)]\label{def:rel-hom}
Two parallel transport maps $\mathrm{PT}_{\cC}$ and $\mathrm{PT}'_{\cC}$ are relatively homotopic with respect to a choice of subgroupoids $\cP_{\Gamma}\subset\cP_{\cD}$ if there is a homotopy of groupoid homomorphisms $g^{t}_{\cD}:\cP_{\cD}\rightarrow G$ such that $g^{t}_{\cD}|_{\cP_{\Gamma}} = \mathrm{PT}_{\Gamma}$ is a fixed LG field $\forall t\in[0,1]$, and $g^{0}_{\cD} = \mathrm{PT}_{\cC}|_{\cP_{\cD}}$, $g^{1}_{\cD} = \mathrm{PT}'_{\cC}|_{\cP_{\cD}}$.
\end{definition}

\begin{lemma}\label{lemma:independence}
Given a cellular network groupoid $\cP_{\Gamma}$, relative homotopy of parallel transport maps is independent of the choice of extending cellular path groupoid $\cP_{\cD}\supset\cP_{\Gamma}$.
\end{lemma}

\begin{proof}
Let $\cD'$ be any other collection of diffeomorphisms with the same cellular network groupoid $\cP_{\Gamma}$. This topological constraint implies that there is a homotopy of diffeomorphisms $\cD^{t}$ between $\cD$ and $\cD'$ preserving $\Gamma$, and an induced homotopy of cellular path groupoids $\cP_{\cD}^{t}$ between $\cP_{\cD}$ and $\cP_{\cD'}$ preserving $\cP_{\Gamma}$. Therefore, two parallel transport maps $\mathrm{PT}_{\cC}$ and $\mathrm{PT}'_{\cC}$ with a common LG field $g:\cP_{\Gamma}\rightarrow G$ are relatively homotopic with respect to $\cP_{\cD}$ if and only if they are relatively homotopic with respect to $\cP_{\cD'}$. 
\end{proof}


\begin{definition}\label{def:ELGT data}
Let  $(M,\cC,\cP_{\Gamma})$ be a cotriangulated manifold with a choice of cellular network $\cP_{\Gamma}$. 
An \emph{extended lattice gauge (ELG) field} is a relative homotopy class of smooth parallel transport maps 
with respect to a fixed LG field $\mathrm{PT}_{\Gamma}:\cP_{\Gamma}\rightarrow G$ and any cellular path subgroupoid $\cP_{\Gamma}\subset\cP_{\cD}\subset\cP_{\cC}$.
\end{definition}

We conclude this section with some remarks on path groupoids and smooth parallel transport maps that will be useful subsequently.

\begin{lemma}\label{lemma:path factorization}
Every element $[\gamma]\in \cP_{\cC}$ admits a local factorization
\[
[\gamma] = \left[\gamma_{v_{r}}\right]\cdot\dots\cdot \left[\gamma_{v_{2}}\right]\cdot\left[\gamma_{v_{1}}\right],
\]
with $\left[\gamma_{v_{i}}\right]\in\cP_{v_{i}}$ for some $c_{v_{1}},\dots,c_{v_{r}}\in\cC_{n}$.  
Therefore, the path groupoid $\cP_{\cC}$ is in particular generated by the local subgroupoids $\{\cP_{v}\}_{c_{v}\in\cC_{n}}$.
\end{lemma}
\begin{proof}
There exist multiple factorizations for a given path, but we prescribe one with a minimality property, as follows.  
Choose any representative $\gamma\in[\gamma]$ whose image intersects the interiors of a finite and  \emph{minimal}  number of $n$-cells $c_{v_{1}},\dots,c_{v_{r}}\in\cC_{n}$ (in the sense that there is no subcollection of $n$-cells with the same property, although a given $n$-cell may appear several times). Then, there exist subintervals $[a_{1},b_{1}],\dots,[a_{r-1},b_{r-1}]$ in $[0,1]$, such that
\[
\gamma([a_{i},b_{i}])\subset \overline{c_{v_{i}}}\cap\overline{c_{v_{i+1}}}, 
\]
while for any other $t\in[0,1]\setminus\left(\cup_{i=1}^{r-1} [a_{i},b_{i}]\right)$, $\gamma(t)$ lies in the interior of one of the previous $n$-cells.
For every $i=1,\dots,r-1$, choose the minimal subcell $c_{\sigma_{i}}$ of $\overline{c_{v_{i}}}\cap\overline{c_{v_{i+1}}}$ such that $\gamma([a_{i},b_{i}])\subset \overline{c_{\sigma_{i}}}$. Then, 
there exists another path $\gamma'$, intimate to $\gamma$, with an additional factorization into $r$ subpaths
\[
\gamma' = \gamma'_{r}\cdot\dots\cdot\gamma'_{1},
\]
such that $\gamma'_{i}([0,1])\subset\overline{c_{v_{i}}}$ for each $i=1,\dots,r$, and $t(\gamma'_{i}) = p_{\sigma_{i}}$  for each $i=1,\dots,r-1$. Letting $[\gamma_{v_{i}}] = [\gamma'_{i}]$, the claim follows.
\end{proof}

It follows from lemma \ref{lemma:path factorization} that a smooth parallel transport map is equivalent to a collection of compatible groupoid homomorphisms satisfying a suitable smoothness condition
\[
\left\{\mathrm{PT}_{\sigma}:\cP_{\sigma}\to G\right\}_{c_{\sigma}\in\cC},
\]
in the sense that for any $c_{\sigma}\subset\overline{c_{\tau}}$ and $[\gamma]\in\cP_{\sigma}\subset \cP_{\tau}$,  $\mathrm{PT}_{\tau}([\gamma]) = \mathrm{PT}_{\sigma}([\gamma])$. We may consider each possibility whenever it is more convenient. The local factorization of elements in $\cP_{\cC}$ motivates local path families as the essential building blocks for the study of local subgroupoids in $\cP_{\cC}$.

\begin{remark}\label{remark:factorization}
The cellular path families satisfy the following property. For any $k''$-subcell $c_{\sigma''}\subset\overline{c_{\sigma}}\cap\overline{c_{\sigma'}}$, $k'' >0$, consider the families $\cF_{\sigma' \sigma}$, $\cF_{\sigma'' \sigma}$, $\cF_{\sigma'' \sigma'}$. Whenever $x\in\overline{c_{\sigma''}}$, a local factorization of intimacy classes is induced
\begin{equation}\label{eq:factorization}
\left[\gamma^{x}_{\sigma' \sigma}\right] = \left[\gamma^{x}_{\sigma'' \sigma'}\right]^{-1} \cdot \left[\gamma^{x}_{\sigma'' \sigma}\right],
\end{equation}
with $\left[\gamma^{x}_{\sigma'' \sigma}\right]\in\cF_{\sigma'' \sigma}$, $\left[\gamma^{x}_{\sigma'' \sigma'}\right]\in\cF_{\sigma'' \sigma'}$, which is written symbolically as 
\[
\cF_{\sigma' \sigma}|_{\overline{c_{\sigma''}}} = \cF_{\sigma'' \sigma'}^{-1}\cdot\cF_{\sigma'' \sigma},
\]
and which in particular applies when $\overline{c_{\sigma''}}\subset\overline{c_{\sigma'}}\subset\overline{c_{\sigma}}$. There are several relevant types of collections of cellular path families. 
Some important ones are given in terms of a collection
\[
\mathfrak{F}_{\text{min},\cD} =\left\{\cF_{v w}\right\}_{c_{v}, c_{w}\in\cC_{n},\, \overline{c_{v}}\cap\overline{c_{w}}\neq \emptyset}.
\]
Hence $\forall x\in\overline{c_{v}}\cap\overline{c_{w}}$, $\gamma_{vw}^{x}\in \cF_{vw}$ and $\gamma_{wv}^{x}\in\cF_{wv}$ are related as $\left[\gamma_{wv}^{x}\right] = \left[\gamma_{vw}^{x}\right]^{-1}$. Any $\cF_{vw}\in\mathfrak{F}_{\text{min},\cD}$ induces an $(n-1)$-dimensional cell in $\cP_{\cC}$ from its bijective correspondence with $\overline{c_{v}}\cap\overline{c_{w}}$.\footnote{On a cotriangulation, if $\overline{c_{v}}\cap\overline{c_{w}} = \overline{c_{\tau}}$ for a pair $c_{v},c_{w}\in\cC_{n}$, then $c_{\tau}\in\cC_{n-1}$ necessarily.} Together, they generate an $(n-1)$-dimensional topological subgroupoid $\cP_{\text{min},\cD}\subset\cP_{\cD}$, the \emph{core} of $\cP_{\cD}$, which is the minimal cellular path subgroupoid containing the families in $\mathfrak{F}_{\text{min},\cD}$. 
\end{remark}


\section{Local and global topological aspects of extended LG fields}\label{sec:cellular data}

We will now assume for simplicity that $G$ is a \emph{connected} Lie group. The extra complications of the general case are easy to sort out, as the connected component of the identity $G_{0}$ is normal in $G$, and $\pi_{0}(G,e)\cong G/G_{0}$. For any $c_{\tau}\in\cC$, a continuous map $h:\overline{c_{\tau}}\to G$ will be called \emph{cellularly-smooth} if its restriction to any subcell $c_{\sigma}\subset \overline{c_{\tau}}$ is smooth. 

\begin{definition}\label{def:glueing map}
A collection of \emph{glueing maps} on $(M,\cC)$ is an assignment of a cellularly-smooth map $g_{\sigma\tau}:\overline{c_{\sigma}}\to G$ for any flag $\overline{c_{\sigma}}\subset\overline{c_{\tau}}$ in $\cC$ satisfying the factorization property
\begin{equation}\label{eq:compatibility-glueing}
\left(g_{\sigma\tau}|_{\overline{c_{\sigma'}}}\right) = g_{\sigma'\sigma}^{-1}\cdot g_{\sigma'\tau}
\end{equation}
whenever $c_{\sigma'}\subset c_{\sigma}\subset c_{\tau}$.
A collection of \emph{clutching maps}, or \emph{skeletal transition functions}, is an assignment of a cellularly-smoot map $h_{vw}:\overline{c_{v}}\cap\overline{c_{w}}\to G$ for any pair $c_{v},c_{w}\in\cC_{n}$ with $\overline{c_{v}}\cap\overline{c_{w}}\neq\emptyset$ satisfying $h_{vw} = h_{wv}^{-1}$ and the cocycle conditions
\[
h_{uw}\vert_{\overline{c_{\sigma}}} = \left(h_{uv}\vert_{\overline{c_{\sigma}}}\right)\left(h_{vw}\vert_{\overline{c_{\sigma}}}\right)\qquad \forall\, c_{\sigma}\in\cC_{n-2},\quad \overline{c_{\sigma}}=\overline{c_{u}}\cap\overline{c_{v}}\cap\overline{c_{w}}.
\]
\end{definition}

\begin{remark}\label{rem:induced-maps}
Every pair $(\mathrm{PT}_{\cC},\cP_{\cD})$ of a smooth parallel transport map and cellular path subgroupoid determines a collection of glueing maps, according to the rule
\begin{equation}\label{eq:glueing PT}
g_{\sigma\tau}(x) = \mathrm{PT}_{\cC}\left(\left[\gamma^{x}_{\sigma\tau}\right]\right),\qquad \gamma^{x}_{\sigma\tau}\in\cF_{\sigma\tau},
\end{equation}
and a collection of clutching maps\footnote{The clutching maps of a pair $(\mathrm{PT}_{\cC}, \cP_{\cD})$ will play the role of transition functions for the principal $G$-bundle $P$ of a gauge field (cf. the parallel transport functions in \cite{PS90,PS93}), and transform as such under homotopies of $\mathrm{PT}_{\cC}$ (proposition \ref{prop:equivalence g}).}
\begin{equation}\label{eq:clutching}
h_{v w}(x) = \mathrm{PT}_{\cC}\left(\left[\gamma^{x}_{vw}\right]\right) = g_{\tau v}^{-1}\cdot g_{\tau w},\qquad \gamma^{x}_{vw}\in\cF_{vw}.
\end{equation}
In particular, the LG field of the pair $(\mathrm{PT}_{\cC},\cP_{\cD})$ is equivalent to the set of values $\{g_{\sigma\tau}(p_{\sigma'})\,|\,c_{\sigma}\subset\overline{c_{\tau}},\, c_{\sigma'}\in \cC_{0,\sigma}\}$, where $\cC_{0,\sigma}:=\left\{c_{\sigma'}\in\cC_{0}\;\vert\; c_{\sigma'}\subset\overline{c_{\sigma}}\right\}$.
\end{remark}

\begin{definition}
For a finite subset $Y\subset X$, where $X \cong S^{k}$ or $\overline{\DD^{k}}$, and a map $g_{0}:Y\to G$, a \emph{ relative homotopy class} $\left[X,G,g_{0}\right]$  is an equivalence class of cellularly-smooth maps $g:X\to G$ via cellularly-smooth homotopies subject to the constraint $g|_{Y} = g_{0}$. 
\end{definition}

\begin{definition}\label{def:bundle data}
Let  $(M,\cC)$ be a cotriangulated $n$-dimensional manifold. A choice of \emph{cellular bundle data} relative to $\cC$ is an equivalence class of collections of clutching maps 
\[
\left[\{h_{v w}:\overline{c_{v}}\cap\overline{c_{w}}\rightarrow G\}_{c_{v},c_{w}\in\cC_{n},\; c_{v}\cap c_{w}\neq\emptyset}\right],
\]
where two collections are equivalent if there exists simultaneous clutching map homotopies $(h_{uv}, h_{vw},h_{wu})\sim(h'_{uv}, h'_{vw},h'_{wu})$ (\emph{cellular equivalences}) for all triples $\{c_{u},c_{v},c_{w}\}\subset\cC_{n}$ such that $\overline{c_{u}}\cap\overline{c_{v}}\cap\overline{c_{w}}\neq\emptyset$.
\end{definition}

\begin{remark}
A choice of cellular bundle data is equivalent to an isomorphism class of principal $G$-bundles. Details are given in theorem \ref{theo:equivalence} (appendix \ref{app:cellular bundle data}). 
\end{remark}

\begin{theorem}\label{theo:hom-class-glueing}
An ELG field is equivalent to a class of collections of gluing maps $\{g_{\sigma\tau}\,|\,c_{\sigma}\subset\overline{c_{\tau}}\}$ under simultaneous relative homotopies with respect to a fixed set of values $\{g_{\sigma\tau}(p_{\sigma'})\,|\,c_{\sigma}\subset\overline{c_{\tau}},\, c_{\sigma'}\in \cC_{0,\sigma}\}$. In particular, to every ELG field there is associated an isomorphism class of principal $G$-bundles on $M$ and a LG field on $(M,\Ga)$.
\end{theorem}

\begin{proof}
By definitions \ref{def:rel-hom}--\ref{def:ELGT data}, lemma \ref{lemma:independence}, and remark \ref{rem:induced-maps}, an ELG field determines a class of collections of gluing maps $\{g_{\sigma\tau}\,|\,c_{\sigma}\subset\overline{c_{\tau}}\}$ under simultaneous relative homotopy with respect to the set $\{g_{\sigma\tau}(p_{\sigma'})\,|\,c_{\sigma}\subset\overline{c_{\tau}},\, c_{\sigma'}\in \cC_{0,\sigma}\}$. Conversely, consider any class of collections $\{g_{\sigma\tau}\,|\,c_{\sigma}\subset\overline{c_{\tau}}\}$ with respect to a fixed set of values at 0-cells. Then the predetermined choice of discrete subgroupoid $\cP_{\Ga}\subset \cP_{\cC}$ determines a LG field $\mathrm{PT}_{\Ga}$ on $M$. The clutching maps $\{h_{vw} =  g_{\tau v}^{-1}\cdot g_{\tau w}\}$ of any collection representing the class induce a principal $G$-bundle $P$ with a set of fiber points $\cE_{\cC}$, whose isomorphism type is an invariant of the class. Any isomorphism type $\{P,\cE_{\cC}\}$ can be prescribed in this way, since a choice of clutching maps representing it can always be deformed locally to have any given values at 0-cells. We will prove the existence of a right-invariant smooth horizontal distribution on $P$ generating a collection of class representatives $\{g_{\sigma\tau}\}$, in terms of right-invariant  lifts of the elements of any auxiliary choice of cellular path subgroupoid $\cP_{\cD}\supset\cP_{\Ga}$. The ELG field of the resulting connection will be independent of the choices of $\cP_{\cD}$ and glueing map representatives. 

By construction, the bundle $P$ is trivialized along the restrictions $P\vert_{\overline{c_{\sigma}}}$ in terms of $\cP_{\cD}$, and the gluing maps $\{g_{\sigma\tau}\}$ compare trivializations whenever $\overline{c_{\sigma}}\subset\overline{c_{\tau}}$. In particular, this implies the existence of right-invariant collections of horizontal lifts of the elements in the smooth families $\{\cF_{\sigma\tau}\}$,
\[
\widetilde{\gamma^{x}_{\sigma\tau}} = \left(\widetilde{\gamma^{x}_{\sigma}}\right)^{-1}\cdot\widetilde{\gamma^{x}_{\tau}}, \qquad \pi\circ \widetilde{\gamma^{x}_{\sigma\tau}} = \gamma^{x}_{\sigma\tau},\qquad \forall\; \overline{c_{\sigma}}\subset\overline{c_{\tau}},\quad x\in \overline{c_{\sigma}},
\] 
in such a way that the parallel transport of the unique lift $\widetilde{\gamma^{x}_{\sigma\tau}}$ with source $s\left(\widetilde{\gamma^{x}_{\sigma\tau}}\right) = b_{\tau}$ equals $g_{\sigma\tau}(x)$ (notice that this is an order 0 constraint on lifts). Since $\forall\; c_{\sigma}\in\cC\setminus \cC_{0}$ the collections of induced horizontal ray lifts $\left\{\widetilde{\gamma^{x}_{\sigma}}\right\}$ determine a foliation of $\pi^{-1}\left(\overline{c_{\sigma}}\setminus\{p_{\sigma}\}\right)$, there is an induced 1-dimensional right-invariant horizontal distribution on $\pi^{-1}\left(\overline{c_{\sigma}}\setminus\{p_{\sigma}\}\right)$, defined in terms of the tangent lines $T_{b}\left(\widetilde{\gamma^{x}_{\sigma}}\right)$ for any $b\neq s\left(\widetilde{\gamma^{x}_{\sigma}}\right)$ in the image of any given induced ray lift $\widetilde{\gamma^{x}_{\sigma}}$. 

In order to show the existence of the desired smooth connection on $P$, additional first order constraints must be imposed by allowing deformations of the previous choices of horizontal lifts, in such a way that their sources and targets remain fixed. These constrains are unobstructed, inductive, and solved by smooth $(n-k+1)$-dimensional right-invariant horizontal distributions $\cK_{\sigma}$ on $\pi^{-1}\left(\overline{c_{\sigma}}\setminus\{p_{\sigma}\}\right)$ for every $c_{\sigma}\in\cC_{k}$, $k\geq 1$, satisfying suitable compatibility conditions. The constraints are given as follows:


\begin{itemize}
\item[(i)]$ \forall\; c_{\sigma_{1}}, c_{\sigma_{2}}\in\cC_{k}$, $\cK_{\sigma_{1}}\vert_{\pi^{-1}(\overline{c_{\sigma_{1}}}\cap \overline{c_{\sigma_{2}}})} = \cK_{\sigma_{2}}\vert_{\pi^{-1}(\overline{c_{\sigma_{1}}}\cap \overline{c_{\sigma_{2}}})}$.

\item[(ii)] $\forall\; \overline{c_{\sigma}}\subset\overline{c_{\tau}}$, $\cK_{\tau}\vert_{\pi^{-1}\left(\overline{c_{\sigma}}\setminus\{p_{\sigma}\}\right)}\subset \cK_{\sigma}$,

\item[(iii)] $\forall\; c_{\sigma}\in \cC$, there is a small open neighborhood $p_{\sigma}\in\cW_{\sigma}$ and a smooth $n$-dimensional right-invariant horizontal distribution $\cJ_{\sigma}$ in $\pi^{-1}\left(\cW_{\sigma}\right)$ such that $\cK_{\sigma}\vert_{\pi^{-1}\left(\cW_{\sigma}\cap c_{\sigma}\setminus\{p_{\sigma}\}\right)}\subset \cJ_{\sigma}\vert_{\pi^{-1}\left(\cW_{\sigma}\cap c_{\sigma}\setminus\{p_{\sigma}\}\right)}$.
\end{itemize}



\noindent  The distributions $\cK_{\sigma}$ are constructed inductively over $n$-cells, $(n-1)$-cells, and so on, by combining the different 1-dimensional distributions constructed before and deforming the lifts accordingly in such a way that the conditions (i)--(iii) are met (to do this, observe that a $k$-cell is in the boundary of exactly $n-k+1$ different $(k+1)$-cells). Any smooth $n$-dimensional right-invariant horizontal distribution on $P$ that contains each $\cK_{\sigma}$ 
is equivalent to a smooth connection on $P$ inducing the collection $\{g_{\sigma\tau}\}$ in terms of $\cP_{\cD}$ and $\cE_{\cC}$.
\end{proof}

\begin{proof}[\textbf{Proof of proposition \ref{prop:interpolation}}]
For any LG field on $(M,\Gamma)$ and any isomorphism class $\{P\}\in\check{H}^{1}(M,G)$, consider any compatible cell decomposition $\cC$, and any set of clutching maps $\{h_{vw}\}$ representing $\{P\}$ and compatible with the corresponding set of values $\{g_{\sigma\tau}(p_{\sigma'})\,|\,c_{\sigma}\subset\overline{c_{\tau}},\, c_{\sigma'}\in \cC_{0,\sigma}\}$. It follows from theorem \ref{theo:hom-class-glueing} that there is an ELG field relative to $(\cC,\Gamma)$ extending $\mathrm{PT}_{\Gamma}$  whose isomorphism class of principal $G$-bundles is $\{P\}$. By theorem \ref{theo:PT}, any representative parallel transport map $\mathrm{PT}_{\cC}$ is equivalent to a gauge orbit of connections with prescribed LG field on each representative $P\in\{P\}$. 
\end{proof}

\begin{remark}\label{rem:surj}
In proving proposition \ref{prop:interpolation}, it is established that the map
\[
\{\mathrm{PT}_{\cC}\} \mapsto (\{\pi: P\rightarrow M\}, \mathrm{PT}_{\Ga}),
\]
assigning to an ELG field its isomorphism class of principal bundles and LG field, is surjective. This is elaborated in section \ref{sec:category}.
\end{remark}

The construction of glueing maps $\left\{g_{\sigma\tau}\right\}$ from a pair $\left(\mathrm{PT}_{\cC},\cP_{\cD}\right)$ leads to another set of relative homotopy data which is actually \emph{equivalent} to an ELG field dissected into local pieces. This fact provides a mechanism to identify when two smooth parallel transport maps yield isomorphic principal $G$-bundles. We present such data in theorem \ref{theo:dissection-data}. The missing step for the new characterization is a glueing mechanism described in lemma \ref{lemma:glueing}.

\begin{lemma}\label{lemma:glueing}
Let $\cC$ be a cotriangulation of the $k$-sphere $S^{k}$ and $g_{0}:\cC_{0}\to G$ be fixed. If every $c_{\tau}\in\cC_{k}$ has assigned a relative homotopy class $[g_{\tau}] \in \left[\overline{c_{\tau}}, G, g_{0}|_{\cC_{0,\tau}}\right]$, such that for any $c_{\tau_{1}},c_{\tau_{2}}\in\cC_{k}$ with $\overline{c_{\sigma}}=\overline{c_{\tau_{1}}}\cap \overline{c_{\tau_{2}}}\neq\emptyset$, we have that $\left[g_{\tau_{1}}|_{\overline{c_{\sigma}}}\right] = \left[g_{\tau_{2}}|_{\overline{c_{\sigma}}}\right]$ as homotopy classes in $\left[\overline{c_{\sigma}},G,g_{0}|_{\cC_{0,\sigma}}\right]$, then the classes $\{[g_{\tau}]\}_{c_{\tau}\in\cC_{k}}$ can be glued into a well-defined element $[g]\in\left[S^{k},G,g_{0}\right]$.
\end{lemma}

\begin{proof}
The fundamental step to prove the lemma is a mechanism to glue neighboring classes 
to a new class over the union of their domain. After repeating the step as many (finite) times as necessary, we obtain an honest homotopy class in $\left[S^{k},G,g_{0}\right]$.  

Denote by $c_{\tau_{1}}$ an arbitrary $k$-cell. Consider all $k$-cells neighboring $c_{\tau_{1}}$, i.e. the $c_{\tau'}\in \cC_{k}$ such that $\overline{c_{\tau_{1}}}\cap\overline{c_{\tau'}}\neq\emptyset$, which can also be labeled as $c_{\tau_{2}},\dots,c_{\tau_{m_{1}}}$. The first step is to construct a class of cellularly equivalent maps defined over $\cup_{i=1}^{m_{1}}\overline{c_{\tau_{i}}}$. Such set is either homeomorphic to $\overline{\DD^{k}}$, to $S^{k}$ minus a finite number of open disks, or to $S^{k}$. We will consider each case, noting that the first one leads to the last two. 

Let us first assume that  $\cup_{i=1}^{m_{1}}\overline{c_{\tau_{i}}}$ is a closed disk in $S^{k}$. By hypothesis, for any pair $\{c_{\tau_{1}},c_{\tau_{i}}\}$ with $\overline{c_{\tau_{1}}}\cap\overline{c_{\tau_{i}}}=\overline{c_{\sigma_{1i}}}$, the induced classes $[g_{\tau_{1}} |_{\overline{c_{\sigma_{1i}}}}]$ and $[g_{\tau_{i}} |_{\overline{c_{\sigma_{1i}}}}]$ coincide.  
 Since there are no local obstructions over multiple intersections, we may choose representatives such that $g_{\tau_{1}} |_{c_{\tau_{1i}}}=g_{\tau_{i}} |_{c_{\tau_{1i}}}$,  
 and moreover, we can apply the same principle to all intersections $\overline{c_{\tau_{ij}}}=\overline{c_{\tau_{i}}}\cap\overline{c_{\tau_{j}}}$, $2\leq i,j\leq m_{1}$, we may also assume that $g_{\tau_{i}} |_{c_{\tau_{ij}}}=g_{\tau_{j}} |_{c_{\tau_{ij}}}$. Hence, the representatives $g_{\tau_{1}},\dots,g_{\tau_{m_{1}}}$ glue to define a piecewise-smooth function on $\cup_{i=1}^{m_{1}}\overline{c_{\tau_{i}}}$, and hence induce a class of cellularly equivalent maps on $\cup_{i=1}^{m_{1}}\overline{c_{\tau_{i}}}$.
We now add all $k$-cells $c_{\tau''}$ intersecting nontrivially with $\cup_{i=1}^{m_{1}}\overline{c_{\tau_{i}}}$, and repeat the previous procedure, until the complementary $k$-cells do not intersect pairwise. The resulting set must be $S^{k}$ minus a finite number of open disks. Call the remaining cells $c_{\tau_{f_{1}}},\dots,c_{\tau_{f_{l}}}$. Upon ordering the closures of $(k-1)$-cells $c_{\sigma}$ in the boundary of each $\overline{c_{\tau_{f_{j}}}}$,  there is a representative $g_{\tau_{f_{j}}}$ in its class whose values at each $\overline{c_{\sigma}}$ coincide with the boundary values of the previously constructed $g$. In this way, we obtain a piecewise-smooth map $g:S^{k}\to G$, up to piecewise-smooth homotopy fixing the values at every 0-cell, whose restriction to any $k$-cell recovers the starting homotopy classes.

To see that the previous procedure is independent of the choice of $k$-cell in $S^{k}$ at every step, repeat it with any other choice of $k$-cells at every step, and call $g'$ the constructed map. Since $g'$ attains the same values than $g$ at any 0-cell, and by hypothesis the restriction of $g' g^{-1}: S^{k}\to G$ to any of the $k$-cells of $\cC$ determines a trivial cellular homotopy class, it follows that the homotopy class of $g' g^{-1}$, as an element in $\pi_{k}(G,e)$, must be trivial. In particular, it follows that $[g']=[g]$ as elements in $\left[S^{k},G,g_{0}\right]$.
\end{proof}

\begin{remark}
Let $c_{\sigma}\in\cC_{k}$ and $g_{0}:\cC_{0,\sigma}\rightarrow G$ a fixed map. If $c_{\sigma'}\subset\overline{c_{\sigma}}$, every element in $[\overline{c_{\sigma}},G,g_{0}]$ induces an element in $\left[\overline{c_{\sigma'}},G,g_{0}|_{\cC_{0,\sigma'}}\right]$. We can give a recursive description of the space $[\overline{c_{\sigma}},G,g_{0}]$ using lemma \ref{lemma:glueing}, in terms of two types of data. The first type is parametrized by the subspace $\cH_{\sigma}$ of 
\[
\bigsqcup_{c_{\sigma'}\subset\overline{c_{\sigma}}} \left[\overline{c_{\sigma'}},G,g_{0}|_{\cC_{0,\sigma'}}\right]
\]
of elements such that over any $\partial\overline{c_{\sigma'}}$, the glueing of homotopy classes from lemma \ref{lemma:glueing} is an element in $\left[S^{k'},G,g_{0}|_{\cC_{0,\sigma'}}\right]$ that is trivial when identified with the free homotopy class that contains it, i.e. every representative is also homotopic  to a constant map, if one forgets the fixed values at 0-cells. We will call $\cH_{\sigma}$ the space of \emph{boundary homotopy data}. In this way, we get a surjection
\[
\mathrm{pr}_{\sigma}:[\overline{c_{\sigma}},G,g_{0}]\to \cH_{\sigma},
\]
whose fibers are all homotopy extension classes from $\partial \overline{c_{\sigma}}$ to $\overline{c_{\sigma}}$, and are a principal homogeneous space for the group $\pi_{k}(G,e)$. Such structure is described in lemma \ref{lemma:torsor}, necessary for the proof of theorem \ref{theo:dissection-data} and corollary \ref{cor:torsor}. Its first part is a standard result used to prove the abelian nature of the homotopy groups of $G$, and we include it for the sake of completeness. 
\end{remark}

\begin{lemma}\label{lemma:torsor}
The groups $\pi_{k}(G,e)$ can be realized as the sets of homotopy classes of maps $g:\overline{\DD^{k}}\to G$ such that $g|_{\partial\overline{\DD^{k}}} = e$, with product induced from pointwise multiplication of maps. For any $k$-cell $c_{\sigma}$ and $g_{0}:\cC_{0,\sigma} \rightarrow G$, there is a free action of $\pi_{k}(G,e)$ on $[\overline{c_{\sigma}},G,g_{0}]$ defined equivalently in terms of left or right pointwise multiplication, whose orbits are the fibers of $\mathrm{pr}_{\sigma}$.
\end{lemma}

\begin{proof}
Let  $x_{1},\dots,x_{k}$ be cartesian coordinates in $\RR^{k}$. The multiplication of two classes $[g], [g']\in\pi_{k}(G,e)$ can be described in terms of the choice of a pair of diffeomorphisms $\psi_{\pm}:\DD^{k}_{\pm}\to \DD^{k}$, where $\DD^{k}_{+}$ (resp. $\DD^{k}_{-}$) is the set of points in $\DD^{k}$ such that $x_{k}> 0$ (resp. $x_{k} < 0$), by letting $[g]\ast [g']$ be the class of maps $g\ast g'$ such that 
\[
g\ast g'|_{\DD^{k}_{+}} = g\circ \psi_{+}\quad \text{and} \quad g\ast g'|_{\DD^{k}_{-}} = g'\circ \psi_{-} 
\]
for any pair of representatives $g,g'$. Upon the choice of a pair of 1-parameter families of open cells $\DD^{k}_{\pm}(t)\subset\DD^{k}$ such that 
\[
\DD^{k}_{\pm}(0) = \DD^{k}_{\pm}\qquad \text{and} \qquad\DD^{k}_{\pm}(1) = \DD^{k}, 
\]
and complemented with a pair of 1-parameter families of diffeomorphisms $\psi_{\pm}^{t}:\DD^{k}_{\pm}(t)\to\DD^{k}$ such that $\psi_{\pm}^{0} = \psi_{\pm}$ and $\psi_{\pm}^{1} = \textrm{Id}$, we can define a based homotopy between any given representative $g\ast g'$ and the pointwise product $gg':\overline{\DD^{k}}\to G$ by letting $g\ast g' (t)$ to be equal to 
\[
\left\{\begin{array}{cl}
g\circ\psi_{+}^{t} & \text{on} \quad\DD^{k}_{+}(t)\setminus \DD^{k}_{+}(t)\cap \DD^{k}_{-}(t),\\\\
 g'\circ\psi_{-}^{t} &  \text{on} \quad\DD^{k}_{-}(t)\setminus \DD^{k}_{-}(t)\cap \DD^{k}_{-}(t),\\\\
 \left(g\circ\psi_{+}^{t}\right) \left(g'\circ\psi_{-}^{t}\right) & \text{on} \quad\DD^{k}_{-}(t)\cap \DD^{k}_{-}(t).
\end{array}\right.
\]
A similar argument shows the existence of a homotopy between $g\ast g'$ and $g'g$. The actions of $\pi_{k}(G,e)$ on $[\overline{c_{\sigma}},G,g_{0}]$ by pointwise left and right multiplication are equivalent by the previous argument, and are free since given any $[g]\in [\overline{c_{\sigma}},G,g_{0}]$, and $[f]\in \pi_{k}(G,e)$, the classes $[fg] = [gf]$ and $[g]$ coincide if and only if $[f]$ is the identity in $\pi_{k}(G,e)$. They also preserve the fibers of $\mathrm{pr}_{\sigma}$, and the induced action on any given fiber of $\mathrm{pr}_{\sigma}$ is transitive.
\end{proof}

\begin{theorem}[Dissection of extended lattice gauge fields]\label{theo:dissection-data}

Every ELG field $\{\mathrm{PT}_{\cC}\}$ on $(M,\cC, \Gamma)$ is equivalent to 
map which assigns, to every flag $\overline{c_{\sigma}}\subset\overline{c_{\tau}}$ in $\cC$, $c_{\sigma}\notin\cC_{0}$, the following collection of local homotopy data of glueing maps:
 
\noindent \emph{(a)}  To every 0-subcell $\overline{c_{\sigma'}}\subseteq \overline{c_{\sigma}}$, a group element 
\[
g_{\sigma\tau}(p_{\sigma'}) = \mathrm{PT}_{\Gamma}\left(\left[\gamma^{p_{\sigma'}}_{\sigma\tau}\right]\right)\in G, 
\]
\noindent \emph{(b)} More generally, to every $k$-subcell $c_{\sigma'}\subseteq \overline{c_{\sigma}}$, $k >0$, an extension class of maps $[g_{\sigma\tau}|_{\overline{c_{\sigma'}}}]$ from $\partial\overline{c_{\sigma'}}$ to $\overline{c_{\sigma'}}$, given the inductive boundary homotopy constraint in $\left[\partial\overline{c_{\sigma'}},G,g_{\sigma\tau}|_{\cC_{0,\sigma'}}\right]$, when $k\geq 2$,\footnote{The sum in the second term of the equality denotes the glueing of homotopy classes of cellularly smooth maps to $\partial\overline{c_{\sigma'}}$ from lemma \ref{lemma:glueing} for the map $g_{\sigma\tau}|_{\cC_{0,\sigma'}}$ of values at 0-cells.}
\begin{equation}\label{eq:boundary-constraint}
\left[g_{\sigma\tau}|_{\partial\overline{c_{\sigma'}}}\right] = \sum_{\{c_{\sigma''}\in\cC_{k-1} \,|\, c_{\sigma''}\subset\overline{c_{\sigma'}}\}} \left[ g_{\sigma\tau}|_{\overline{c_{\sigma''}}}\right], 
\end{equation}
which is trivial when identified with the free homotopy class that contains it. 
 Moreover, whenever $c_{\sigma'}\subset \overline{c_{\sigma}}\subset \overline{c_{\tau}}$, the compatibility conditions
\[
g_{\sigma'\tau}(p_{\sigma''}) = g_{\sigma' \sigma}(p_{\sigma''})\cdot g_{\sigma \tau}(p_{\sigma''}),
\]
are satisfied at every 0-cell $c_{\sigma''}\subset \overline{c_{\sigma'}}$,  and more generally, the compatibility condition
\begin{equation}\label{eq:comp-cond-ext}
\left[g_{\sigma'\tau} |_{\overline{c_{\sigma''}}}\right] = \left[\left(g_{\sigma'\sigma} |_{\overline{c_{\sigma''}}}\right)\cdot \left(g_{\sigma\tau} |_{\overline{c_{\sigma''}}}\right)\right],\footnote{The right-hand side denotes the class of products of piecewise-smooth representatives in the relative extension classes $\left[\left(g_{\sigma'\sigma} |_{\overline{c_{\sigma''}}}\right)\right]$, $\left[\left(g_{\sigma\tau} |_{\overline{c_{\sigma''}}}\right)\right]$, which coincides with the corresponding relative extension class of any class representative (cf. lemma \ref{lemma:torsor}).}
\end{equation}
is satisfied at every $k$-cell $c_{\sigma''}\subseteq\overline{c_{\sigma'}}$, $k=1,\dots,\dim(c_{\sigma'})$. 
\end{theorem}

\begin{proof}
The disection of $\{\mathrm{PT}_{\cC}\}$ as relative extension homotopy classes with fixed values over 0-cells is straightforward.
The restriction of the class representatives $g_{\sigma\tau}$ to any subcell $\overline{c_{\sigma'}}$ determines local homotopy classes of maps over $\overline{c_{\sigma'}}$ relative to the 0-cells in its boundary, that can be factored as a boundary homotopy class over $\partial\overline{c_{\sigma'}}$ which is necessarily trivial as a free homotopy class, together with an extension class to the interior $c_{\sigma'}$. The classes associated to $(k-1)$-cells $\overline{c_{\sigma''}}$ in $\partial\overline{c_{\sigma'}}$, for any given $k$-cell $c_{\sigma'}$, glue according to lemma  \ref{lemma:glueing} to class of maps over $\partial\overline{c_{\sigma'}}$ that are homotopic to a constant map. The homotopical data obtained this way satisfies the compatibility conditions stated above. 

Conversely, given a collection of homotopy classes of extension maps to $c_{\sigma''}$ for every flag $\overline{c_{\sigma''}}\subset\overline{c_{\sigma}}\subset\overline{c_{\tau}}$, with fixed values at 0-cells and satisfying the glueing compatibility conditions above, it follows from lemma \ref{lemma:glueing} that we can reconstruct a class of collections of glueing maps  $\{\{g_{\sigma\tau}\}_{\overline{c_{\sigma}}\subset\overline{c_{\tau}}}\}$, under relative homotopy for the fixed values at 0-cells, as there are no obstructions to extending a given map to the interior of a closed cell. Such a class of collections of glueing maps is equivalent to an ELG field by theorem \ref{theo:hom-class-glueing}.
\end{proof}

The restrictions $\left\{\mathrm{PT}_{\cC}|_{\overline{c_{\sigma}}}\,|\, c_{\sigma}\in\cC\right\}$ of an ELG field are intrinsically \emph{local}, as parallel transport evaluations of the field along the plaquettes contained in $\overline{c_{\sigma}}$ are. As an important consequence, a given ELG field determines a given bundle structure \emph{locally}, in the sense that if its restrictions are known for every $n$-cell in $\cC$, the ELG field (and in particular, its principal $G$-bundle) can be recovered. Hence, any local action functional defined on ELG fields could be expressed as a sum of functionals of restrictions over $n$-cells. 

Consider the core subgroupoid $\cP_{\cD,\text{min}}\subset \cP_{\cD}$ generated by the cellular path subfamilies $\cF_{vw}$ for $c_{v},c_{w}\in\cC_{n}$, $\overline{c_{v}}\cap\overline{c_{w}}\neq\emptyset$.  With it we can define bigger homotopy classes of smooth parallel transport maps $\{\mathrm{PT}_{\cC}\}_{\min}$ with fixed values over the discrete subgroupoid $\cP_{\Gamma,\min} = \cP_{\Gamma}\cap \cP_{\cD,\min}$. Denote the projection of an ELG field $\{\mathrm{PT}_{\cC}\}$ onto the larger class containing it by $\textrm{pr}_{\text{core}}(\{\mathrm{PT}_{\cC}\})$.

\begin{definition}
The \emph{core} of an extended lattice gauge field $\{\mathrm{PT}_{\cC}\}$ is the resulting homotopy class of smooth parallel transport maps with respect to any core subgroupoid $\cP_{\cD,\min}$, i.e. $\{\mathrm{PT}_{\cC}\}_{\min} := \mathrm{pr}_{\text{core}}(\{\mathrm{PT}_{\cC}\})$.
\end{definition}

The core of an extended lattice gauge field constitutes the minimal local homotopical information that is required to recover an isomorphism class of principal $G$-bundles.
On a cotriangulation $\cC$, every $c_{\sigma} \in\cC_{n-2}$ determines 	two sets of triples $\{c_{v_{1}}, c_{v_{2}},c_{v_{3}}\}\subset\cC_{n}$, $\{c_{\tau_{1}},c_{\tau_{2}}, c_{\tau_{3}}\} \subset \cC_{n-1}$  such that $\overline{c_{\sigma}} = \overline{c_{v_{1}}}\cap\overline{c_{v_{2}}}\cap\overline{c_{v_{3}}} = \overline{c_{\tau_{1}}}\cap\overline{c_{\tau_{2}}}\cap\overline{c_{\tau_{3}}}$, together with a system of gapless flags
\begin{equation}\label{eq:cyclic triples}
\begin{matrix}
&&&&\overline{c_{v_{1}}}&&&&\\
&&&\rotatebox[origin=c]{30}{$\subset$}&&\rotatebox[origin=c]{-30}{$\supset$}&&&\\
&&\overline{c_{\tau_{3}}}&\supset&\overline{c_{\sigma}}&\subset&\overline{c_{\tau_{2}}}&&\\
&\rotatebox[origin=c]{30}{$\supset$}&&&\rotatebox[origin=c]{-90}{$\subset$}&&&\rotatebox[origin=c]{-30}{$\subset$}&\\
\overline{c_{v_{2}}}&&\supset &&\overline{c_{\tau_{1}}}&&\subset&&\overline{c_{v_{3}}}\\
\end{matrix}
\end{equation}
For every triple of $n$-cells as above, the cocycle condition for the restriction of the triple of clutching maps to their common $(n-2)$-cell closure $\overline{c_{\sigma}}$ can be formulated as their equivariance under suitable $\mathrm{S}_{3}$-actions, as follows. There is an action of $\mathrm{S_{3}}$ on $G \times G \times G$, prescribed on a choice of two generators in the following way: for the (even) 3-cycle $(123)$, $(123)\cdot (g_{1},g_{2},g_{3}) = (g_{3},g_{1},g_{2})$, and for the (odd) transposition $(12)$, $(12)\cdot (g_{1},g_{2},g_{3}) = (g_{2}^{-1},g_{1}^{-1},g_{3}^{-1})$.
Consider the multiplication map
\[
T_{G}:G \times G \times G \to G,\qquad (g_{1},g_{2},g_{3})\mapsto g_{1}g_{2}g_{3},
\]
let $V_{G}=T^{-1}_{G}(e) $, and denote by $\mathrm{pr}_{i}$ the projections from $G \times G \times G$ into the $i$th factor. 
It follows that $V_{G}$ is invariant under the $\mathrm{S}_{3}$-action defined before, and that $V_{G} \cong G\times G$ under the projections $\mathrm{pr}_{i}\times \mathrm{pr}_{j}$ to any pair of distinct $i,j\in\{1,2,3\}$. The previous action, which we call \emph{triadic}, is faithful. 

Consider a triple of clutching maps for any triple of $n$-cells $\{c_{v_{1}},c_{v_{2}},c_{v_{3}}\}$ as before. 
Define the maps 
\[
\mathbf{h}^{ijk}_{\sigma} := \left(h_{v_{i}v_{j}}|_{\overline{c_{\sigma}}}\right)\times \left(h_{v_{j}v_{k}}|_{\overline{c_{\sigma}}}\right)\times \left(h_{v_{k}v_{i}}|_{\overline{c_{\sigma}}}\right):\overline{c_{\sigma}}\times \overline{c_{\sigma}}\times \overline{c_{\sigma}} \to G\times G\times G
\]
whose restriction to the diagonal $\Delta(\overline{c_{\sigma}}\times\overline{c_{\sigma}}\times\overline{c_{\sigma}})$ lies in $V_{G}$.
The arbitrariness of the labeling in the triple of $n$-cells is equivalent to the fact that the $\mathrm{S}_{3}$-action on the set of maps $\left\{\mathbf{h}^{ijk}_{\sigma}\right\}$
by index permutations, is given by postcomposition with the triadic action on their respective images. 

\begin{corollary}\label{cor:core}
The core $\{\mathrm{PT}_{\cC}\}_{\min}$ of an extended lattice gauge field $\{\mathrm{PT}_{\cC}\}$ is equivalent to a collection of the following local homotopy data:

\noindent \emph{(a)} To every triple of elements $c_{v_{1}}, c_{v_{2}}, c_{v_{3}}\in\cC_{n}$ such that   $\overline{c_{v_{1}}}\cap\overline{c_{v_{2}}}\cap \overline{c_{v_{3}}} = \overline{c_{\sigma}}$
with $c_{\sigma}\in\cC_{n-2}$ as in \eqref{eq:cyclic triples}, we assign
\begin{itemize}
\item[(i)] For every 0-cell $c_{\sigma''_{0}}\subset\overline{c_{\sigma}}$, a point $\mathbf{h}_{\sigma}(c_{\sigma''_{0}})\in V_{G}$,
\item[(ii)] More generally, for every $k$-cell $c_{\sigma'}\subseteq \overline{c_{\sigma}}$, $k=1,\dots,n-2$, a relative homotopy class of maps 
\[
\left[\mathbf{h}_{\sigma}|_{\overline{c_{\sigma'}}}:\overline{c_{\sigma'}}\to V_{G}\right], 
\]
with fixed values over 0-subcells, representing an extension class  from $\partial \overline{c_{\sigma'}}$ to $\overline{c_{\sigma'}}$, that is determined by the inductive boundary data in $\left[\partial \overline{c_{\sigma'}},V_{G},\mathbf{h}_{\sigma}|_{\cC_{0,\sigma'}}\right]$ when $k\geq 2$,
\[
\left[\mathbf{h}_{\sigma}|_{\partial\overline{c_{\sigma'}}}\right] =\sum_{\{c_{\sigma''}\in\cC_{k-1} \,|\, c_{\sigma''}\subset\overline{c_{\sigma'}}\}} \left[\mathbf{h}_{\sigma}|_{\overline{c_{\sigma''}}}\right],
\footnote{The sum denotes the glueing of homotopy classes of cellularly smooth maps constructed in lemma \ref{lemma:glueing}, while $\mathbf{h}_{\sigma}|_{\cC_{0,\sigma'}}$ is the map of prescribed values at 0-cells. Following lemma \ref{lemma:torsor}, the set of such extension classes is a principal homogeneous space for $\pi_{k}(V_{G},\textbf{e})$, where $\textbf{e} = (e,e,e)$ (cf. corollary \ref{cor:torsor}).}
\]
which is trivial when identified with the free homotopy class that contains it. 
\end{itemize}
The assignment is equivariant for the permutation action of the group $\mathrm{S_{3}}$ on the triple $c_{v_{1}}, c_{v_{2}}, c_{v_{3}}$ and its triadic action on  $V_{G}$. 

\noindent \emph{(b)} To every pair of elements  $c_{v}, c_{w}\in\cC_{n}$ such that $\overline{c_{v}}\cap\overline{c_{w}} = \overline{c_{\tau}}$ with $c_{\tau}\in\cC_{n-1}$,  we assign an extension class  $[h_{v w}:\overline{c_{\tau}}\to G]$ of the inductively glued boundary class over $\partial\overline{c_{\tau}}$ 
 to $\overline{c_{\tau}}$, in such a way that the resulting induced glued class $[h_{w v}\cdot h_{v w}]$ is homotopically trivial.
\end{corollary}

\begin{proof}
Follows from theorem \ref{theo:dissection-data} and the relation \eqref{eq:clutching} between clutching and glueing maps, as there is an the equivalence between the core of a ELG field and a homotopy class of collections of clutching maps $\{\{h_{vw}\}_{\overline{c_{v}}\cap\overline{c_{w}}\neq\emptyset}\}$ with fixed values over 0-cells. When the latter are dissected over compatible triples, we obtain the local homotopy data stated above, the first part arising from any triple of cells $\overline{c_{v_{1}}}\cap\overline{c_{v_{2}}}\cap \overline{c_{v_{3}}} = \overline{c_{\sigma}}$, while the residual data takes the form of relative extension classes over all $(n-1)$-cells. Conversely, given any such choice of local homotopy data, it follows from the compatibility conditions, together with lemma \ref{lemma:glueing}, that the homotopy classes induced by the projections $\pi_{i}(\mathbf{h}_{\sigma}|_{\overline{c_{\sigma''}}})$ of the maps $\mathbf{h}_{\sigma}|_{\overline{c_{\sigma''}}}$ into components, with $i=1,2,3$, together with the extension classes $[h_{vw}]$, determine a homotopy class of collections of  clutching maps with fixed values at 0-cells.
\end{proof}

\section{Extended LG fields in small dimensions}\label{sec:small-dim}

To suplement section \ref {sec:cellular data}, we provide the explicit characterization of ELG fields following from theorem \ref{theo:dissection-data} and corollary \ref{cor:core}, for manifolds of dimensions 2, 3 and 4. The characterization suggests a recursive algorithm that could be implemented in general to determine a set of generators for the extension classes, based on a systematic use of the compatibility conditions \eqref{eq:comp-cond-ext}. Recall that for any Lie group $G$, the group $\pi_{1}(G, e)$ is abelian, $\pi_{2}(G,e)$ is trivial, while $\pi_{3}(G, e)$ is torsion-free, and hence isomorphic to $\ZZ^{m}$ for some $m$ \cite{Mil63}.

$\bullet$ ($n=2$)
This is the simplest nontrivial case. The cotriangulations are cell decompositions with exactly 3 edges merging at each vertex \eqref{eq:cyclic triples}, such as the tetrahedral, cubical and dodecahedral cellular representations of the 2-sphere. The only flags that need to be considered take the form $\overline{c_{\tau}}\subset\overline{c_{v}}$ with $c_{\tau}\in\cC_{1}$, and an ELG field is simply a splitting of its core: the latter corresponds to (i) an assignment of a parallel transporter to every path of the form $\left[\gamma_{\sigma v}\right]^{-1}\cdot\left[\gamma_{\sigma w}\right]$ for any pair of 2-cells $\overline{c_{v}},\overline{c_{w}}$ sharing a common boundary $\overline{c_{\tau}}$ containing a 0-cell $c_{\sigma}$, and (ii) a collection of extension classes $[h_{vw}]$ relative to their boundary values, grouping together as the points $\{\mathbf{h}_{\sigma}\} \subset V_{G}$. In turn, an ELG field corresponds to (i) an assignment of a parallel transporter to every path of the form $\left[\gamma_{\sigma \tau}\right]^{-1}\cdot\left[\gamma_{\sigma v}\right]$,\footnote{Observe that $\left[\gamma_{\sigma v}\right]^{-1}\cdot\left[\gamma_{\sigma w}\right] = \left(\left[\gamma_{\sigma \tau}\right]^{-1}\cdot\left[\gamma_{\sigma v}\right]\right)^{-1}\cdot\left(\left[\gamma_{\sigma \tau}\right]^{-1}\cdot\left[\gamma_{\sigma w}\right]\right)$.} and (ii) a collection of extension classes $[g_{\tau v}]$ relative to their fixed (and compatible) boundary values $g_{\tau v}(p_{\sigma})$, $g_{\tau v}(p_{\sigma'})$, $c_{\sigma},c_{\sigma'}\subset\partial \overline{c_{\tau}}$. Hence, every class representative $h_{vw}$ splits as $g_{\tau v}^{-1}\cdot g_{\tau w}$, for a pair of class representatives $g_{\tau v}$, $g_{\tau w}$. All of such extension classes can be parametrized by elements in $\pi_{1}(G,e)$ once auxiliary choices of extension maps are made. The boundary homotopy constraints \eqref{eq:boundary-constraint} and the compatibility conditions \eqref{eq:comp-cond-ext} are vacuous.
 
 \begin{figure}[!ht]
\caption{Stereographic projection of a tetrahedral cellular decomposition of the 2-sphere, together with a cellular network. Base points are indicated with the letter $p$. The 2-cells are labelled with the subscripts $v_{i}$, the 1-cells with $\tau_{j}$, and the 0-cells with $\sigma_{k}$.}\label{sphere}


\centering
\includegraphics[width=0.5\textwidth]{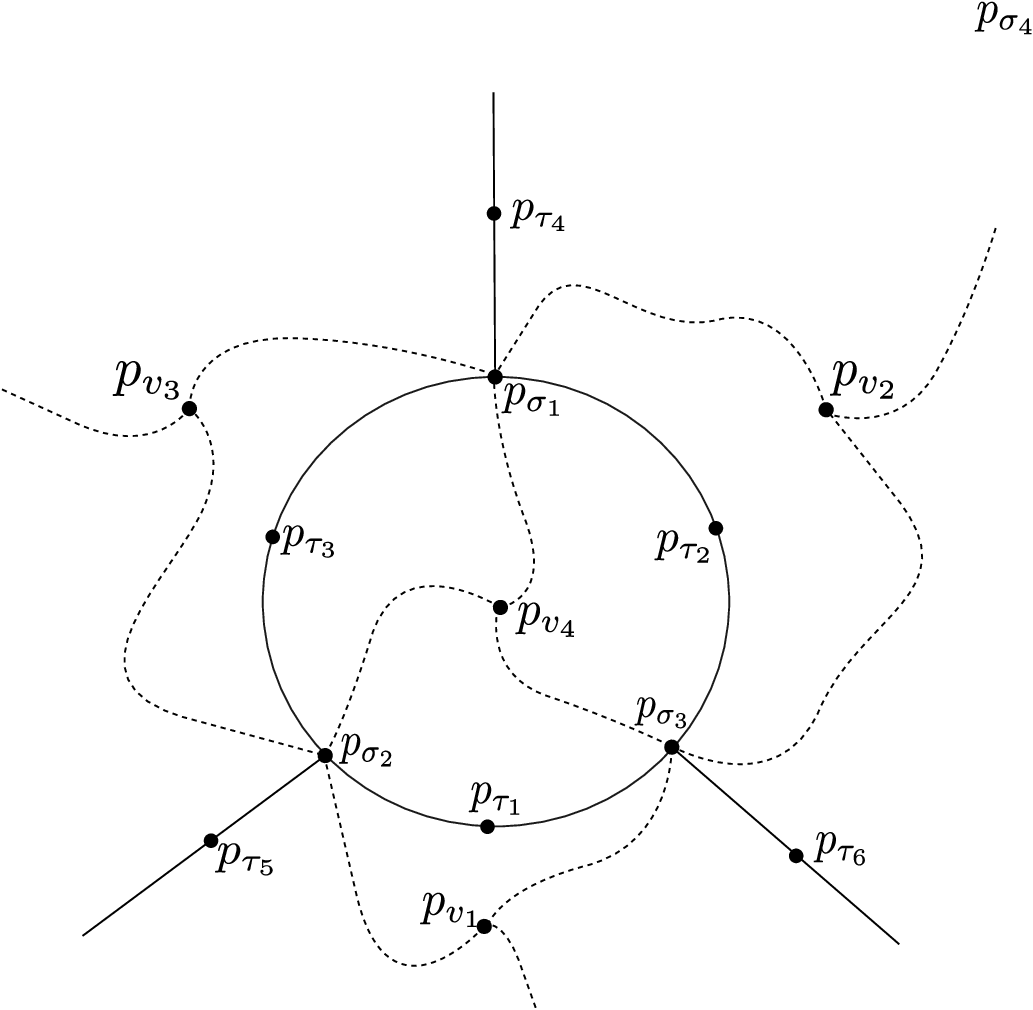}
\end{figure}

In fact, more can be said in dimension 2, regarding the projection of an ELG field to its cellular bundle data. We can prescribe an equivalence class of principal $G$-bundles by means of a ``canonical form" under cellular equivalence of an ELG field. Recall that the equivalence classes of principal $G$-bundles over a surface $S$ are parametrized by $\pi_{1}(G,e)$, and correspond to homotopy classes of transition functions for a trivialization $\{S\setminus\{p\},\cU\}$, where $p$ is an arbitrary point in $S$ and $\cU$ is a small disk containing $p$, after a retraction  from $(S\setminus\{p\})\cap \cU$ to a choice of some simply closed loop $\gamma\subset (S\setminus\{p\})\cap \cU$ is made. If we let $\gamma$ be the boundary of a 2-cell $c_{v}\in\cC$, $p=p_{v}$, and $\cU=\cU_{v}$, we can recover all equivalence classes by choosing any 1-cell $c_{\tau}\subset \partial\overline{c_{v}}$ (such that $\overline{c_{\tau}} = \overline{c_{v}}\cap\overline{c_{w}}$), choosing arbitrary elements in $V_{G}$ for every $c_{\sigma}\in\cC_{0}$, and decreeing all extension classes of clutching maps to be trivial, except for $[h_{v w}]$. The correspondence of such data and equivalence classes of principal $G$-bundles follows when the extension classes of clutching maps over the 1-subcells in $\partial\overline{c_{v}}$ are glued.

$\bullet$ ($n=3$) Over a 3-manifold, besides the prescription of a LG field with generators corresponding to the parallel transports of all paths $\left[\gamma_{\sigma'_{0}\sigma}\right]$ joining the base point of a $k$-cell $c_{\sigma}$, $k = 1,2,3$, and a 0-cell $c_{\sigma'_{0}}$ in its boundary, the compatible values at 0-cells $c_{\sigma'_{0}}$ for the relative homotopy classes of glueing maps can be constructed as
\[
g_{\sigma\tau}\left(p_{\sigma'_{0}}\right) = \mathrm{PT}\left(\left[\gamma_{\sigma'_{0}\sigma}\right]^{-1}\cdot\left[\gamma_{\sigma'_{0}\tau}\right]\right) = \mathrm{PT}\left(\left[\gamma^{p_{\sigma'_{0}}}_{\sigma\tau}\right]\right).
\]
Then, an ELG field can be entirely described in terms of the flags of the form $\overline{c_{\sigma}}\subset\overline{c_{\tau}}$ and $\overline{c_{\sigma}}\subset\overline{c_{v}}$, with $c_{\sigma}\in\cC_{1}$, $c_{\tau}\in\cC_{2}$, $c_{v}\in\cC_{3}$. This is so since the relative homotopy classes $[g_{\sigma\tau}]$, $[g_{\sigma v}]$ are determined by a single extension to $\overline{c_{\sigma}}$, which can be parametrized by elements in the group $\pi_{1}(G,e)$ after an auxiliary choice of extension maps is made. Then, the extensions of a class $[g_{\tau v}]$ over the boundary 1-subcells $\overline{c_{\sigma}}$ are determined recursively from the factorization property of representatives following \eqref{eq:comp-cond-ext}
\[
g_{\tau v}|_{\overline{c_{\sigma}}} = g_{\sigma\tau}^{-1}\cdot g_{\sigma v},
\]
and their glueing over $\partial\overline{c_{\tau}}$ is homotopic to a point in $G$, according to \eqref{eq:boundary-constraint}. The remaining extension class for $[g_{\tau v}]$ to the 2-cell $c_{\tau}$ is necessarily trivial, as a consequence of the triviality of $\pi_{2}(G,e)$.

$\bullet$ ($n=4$) This case is also easy to describe. Let $c_{\sigma'}\in\cC_{1}$, $c_{\sigma}\in\cC_{2}$, $c_{\tau}\in\cC_{3}$, $c_{v}\in\cC_{4}$. Once a LG field has been prescribed in terms of a set of generators as in the case $n=3$,  and the induced choice of all compatible values at 0-cells for the relative homotopy classes of glueing maps is constructed, we proceed as follows. To construct the extension classes, the fundamental scaffolding is determined by the relative homotopy classes $[g_{\sigma' \sigma}]$, $[g_{\sigma' \tau}]$ and $[g_{\sigma' v}]$, which are determined by an extension class from $\partial\overline{c_{\sigma'}}$ to $\overline{c_{\sigma'}}$. As before, such extension classes can be parametrized by elements in $\pi_{1}(G,e)$, once an auxiliary choice of extension map is made. The homotopy classes $[g_{\sigma\tau}|_{\overline{c_{\sigma'}}}]$ and $[g_{\tau v}|_{\overline{c_{\sigma'}}}]$ are determined once again from the factorization property of representatives
\[
g_{\sigma \tau}|_{\overline{c_{\sigma'}}} = g_{\sigma'\sigma}^{-1}\cdot g_{\sigma' \tau},\qquad g_{\tau v}|_{\overline{c_{\sigma'}}} = g_{\sigma'\tau}^{-1}\cdot g_{\sigma' v},
\]
which, in turn, determine the classes $[g_{\sigma v}|_{\overline{c_{\sigma'}}}]$. Once the boundary homotopy constraint \eqref{eq:boundary-constraint} is imposed on the glued classes $[g_{\sigma \tau}|_{\partial\overline{c_{\sigma}}}]$ and $[g_{\tau v}|_{\partial\overline{c_{\sigma}}}]$ (which imply the same constraint for the class $[g_{\sigma v}|_{\partial\overline{c_{\sigma}}}]$), their extensions from $\partial\overline{c_{\sigma}}$ to $\overline{c_{\sigma}}$ are trivial since $\pi_{2}(G,e)$ is so. For the same reason, the boundary homotopy constraint for $[g_{\tau v}|_{\partial\overline{c_{\tau}}}]$  is automatically verified. To conclude, the prescription of the extensions of $[g_{\tau v}]$ from $\partial\overline{c_{\tau}}$ to $\overline{c_{\tau}}$ can be parametrized by $\pi_{3}(G,e)$ once an auxiliary choice of extension map is made. 

A minimal example that holds in the previous cases, and in fact, for arbitrary dimensions, is the $n$-sphere, with a cellular decomposition induced from its realization as the boundary of an $(n+1)$-simplex or an $(n+1)$-cube. For instance, when $n=2$, and $S^{2}$ corresponds to the boundary of a tetrahedron (figure \ref{sphere}), we can list a pair of paths for every 1-cell $c_{\tau}$, joining $p_{\tau}$ and the 0-cells in its boundary $c_{\sigma_{1}},c_{\sigma_{2}}$. Each of these paths gets assigned a corresponding parallel transport value $\mathrm{PT}\left(\left[\gamma_{\sigma_{i}\tau}\right]\right)$, playing the role of a prototypical LG field in the usual sense. The previous paths get complemented with the paths $\left[\gamma_{\sigma'_{0}\sigma}\right]$ in a cellular network, joining the base point of an arbitrary $k$-cell $c_{\sigma}$, $k\geq 2$, to a 0-cell $c_{\sigma'_{0}}$ in its boundary, and also get assigned a parallel transport value. Together, they fully determining a LG field.  The prescription gets completed with an extension class of maps for every flag $\overline{c_{\sigma}}\subset\overline{c_{\tau}}$, and every $k$-subcell $\overline{c_{\sigma'}}\subseteq\overline{c_{\sigma}}$. When $n = 2$, there is one such extension class for every flag $\overline{c_{\tau}}\subset \overline{c_{v}}$ with $c_{\tau}\in\cC_{1}$, and the extensions for the pairs $\overline{c_{\tau}}\subset \overline{c_{v}}$ and $\overline{c_{\tau}}\subset \overline{c_{w}}$, such that $\overline{c_{\tau}} = \overline{c_{v}}\cap\overline{c_{w}}$, are inverse to each other.

\begin{remark}\label{remark:simple-simply-conn}
In the case when $G$ is compact, simple, and simply connected, as is the case for the groups $\mathrm{SU}(r)$ of physical interest, the only relevant homotopical information in small dimensions is contained in the group $\pi_{3}(G,e) \cong \ZZ$ \cite{Bott56,Mil63}. Hence, we conclude \emph{a posteriori} that in such case, LG and ELG fields coincide when $n=2,3$. This observation may look deceiving, but it is actually a confirmation of the fundamental difference in nature between these dimensions and the case $n = 4$, while also confirms that for the former, standard LGT methods are in fact \emph{complete} from a topological perspective. We plan to return to the problem of the determination of the $\mathrm{SU}(r)$-topological charge of an extended lattice gauge field when $n = 4$, and contrast it with the work of Phillips-Stone \cite{PS86,PS90,PS93}.
 \end{remark}

\section{The spaces of extended LG fields}\label{sec:category}

We will now study the spaces of ELG fields for a given Lie group $G$ and a cotriangulation $\cC$ in $M$. Such spaces are smooth finite-dimensional manifolds, whose topology encodes $\Check{H}^{1}\left(M,G\right)$.

Given a triple $(M,\cC,\Gamma)$, let us denote its space of LG fields by $\cM_{\Gamma}$. $\cM_{\Gamma}$ is isomorphic to the Lie group $G^{N_{1}}$, where $N_{1}$ is the number of edges in $\Gamma$, inheriting a smooth manifold structure. Denote the corresponding set of ELG fields by $\cM_{\cC}$. There is an obvious forgetful map 
\begin{equation}\label{eq:covering1}
\mathrm{pr}_{\cC}:\cM_{\cC} \to \cM_{\Gamma}. 
\end{equation}
which is surjective as a consequence of theorem \ref{theo:hom-class-glueing} (see remark \ref{rem:surj}).
In a similar way, if we let  $\cN_{\cC}$ be the set of cores of ELG fields, and $\cN_{\Gamma}$ the smooth manifold of their network data, there is a surjective forgetful map
\begin{equation}\label{eq:covering2}
\mathrm{pr}_{\min}:\cN_{\cC}\to \cN_{\Gamma},
\end{equation}
which together with the additional projections $\mathrm{pr}_{\text{core}}:\cM_{\cC} \to \cN_{\cC},$ and $\mathrm{pr}_{\Gamma,\min}:\cM_{\Gamma} \to \cN_{\Gamma}$ form a commutative diagram

\begin{equation}\label{eq:comm-diag}
\begin{CD}
\cM_{\cC} @>\mathrm{pr}_{\text{core}}>> \cN_{\cC}\\
@V\mathrm{pr}_{\cC}VV @VV\mathrm{pr}_{\min}V\\
\cM_{\Gamma} @>\mathrm{pr}_{\Gamma,\min}>> \cN_{\Gamma}
\end{CD}
\end{equation}
There exist smooth manifold structures on $\cM_{\cC}$ and $\cN_{\cC}$ in such a way that \eqref{eq:covering1} and \eqref{eq:covering2} are covering maps. Such structures are encoded in a ``path lifting" property implicit in the definition of ELG fields. Consider any smooth path $\mathrm{PT}^{t}_{\Gamma}:[0,1]\rightarrow \cM_{\Ga}$  joining any two given LG fields $\mathrm{PT}^{0}_{\Gamma}$ and $\mathrm{PT}^{1}_{\Gamma}$. Then it follows from the inductive mechanism described in theorem \ref{theo:dissection-data} that for every choice of ELG field 
\[
\left\{\mathrm{PT}^{0}_{\cC}\right\}\in \mathrm{pr}_{\cC}^{-1}\left(\mathrm{PT}^{0}_{\Gamma}\right), 
\]
there is an induced homotopy lifting 
\[
\widetilde{\mathrm{PT}^{t}_{\Gamma}}:[0,1]\rightarrow \cM_{\cC} 
\]
whose target $\{\mathrm{PT}^{1}_{\cC}\}:= t\left(\widetilde{\mathrm{PT}^{t}_{\Gamma}}\right)$ lies in $\mathrm{pr}_{\cC}^{-1}\left(\mathrm{PT}^{1}_{\Gamma}\right)$.
Moreover, for any LG field $\mathrm{PT}_{\Gamma}$ in a sufficiently small neighborhood $\cU \subset\cM_{\Ga}$ of $\mathrm{PT}^{0}_{\Gamma}$, its lift $\left\{\mathrm{PT}_{\cC}\right\}\in\cM_{\cC}$ would be independent of the choice of path joining $\mathrm{PT}_{\Gamma}$ and $\mathrm{PT}^{0}_{\Gamma}$, defining a local bijection over $\cU$. Consequently, there is an induced smooth manifold structure on $\cM_{\cC}$. An analogous construction on the points of $\cN_{\cC}$ can be implemented as well. The path lifting property on $\cM_{\cC}$ and $\cN_{\cC}$ implies that the maps \eqref{eq:covering1} and \eqref{eq:covering2} are covering maps. 

The groups of deck transformations of $\cM_{\cC}$ and $\cN_{\cC}$ admit an explicit description in terms of ``generators and relations". Consider the group
\[
\widetilde{G}_{\cC} = \prod_{\overline{c_{\sigma}}\subset\overline{c_{\tau}}}\left(\prod_{\overline{c_{\sigma'}}\subseteq \overline{c_{\sigma}}} G_{\sigma'}\right),
\]
where $G_{\sigma'} = \pi_{k}(G,e)$ for $c_{\sigma'}\in \cC_{k}$ (in particular, $G_{\sigma'} = \{e\}$ when $k = 0$). We will consider a subgroup determined by a series of group relations of two kinds. For any flag $\overline{c_{\sigma'}}\subset\overline{c_{\sigma}}\subset\overline{c_{\tau}}$, and $\overline{c_{\sigma''}}\subseteq\overline{c_{\sigma'}}$, consider the homomorphisms 
\[
\Theta^{\sigma''}_{\sigma'\sigma\tau}:  \widetilde{G}_{\cC}\times  \widetilde{G}_{\cC}\times \widetilde{G}_{\cC}\to G_{\sigma''},
\qquad
\Theta^{\sigma''}_{\sigma'\sigma\tau} := \left(\mathrm{pr}_{\sigma''\sigma'\sigma}\right)\left(\mathrm{pr}_{\sigma''\sigma\tau}\right)\left(\mathrm{pr}_{\sigma''\sigma'\tau}\right)^{-1}
\]
where for an arbitrary $\overline{c_{\sigma'}}\subseteq \overline{c_{\sigma}} \subset \overline{c_{\tau}}$, $\mathrm{pr}_{\sigma'\sigma\tau}$ denotes the projection from $\widetilde{G}_{\cC}$ to the corresponding factor $G_{\sigma'}$ for the flag $\overline{c_{\sigma}}\subset\overline{c_{\tau}}$.
The kernels of these maps among homotopy groups give relations that we will call \emph{multiplicative relations in homotopy}. The second type of relations are described as follows.
For any $(k+1)$-cell $c_{\sigma}$, $k = 1,\dots, n-2$, define the product homomorphism
\[
\alpha_{\sigma} : \prod_{c_{\sigma'}^{k}\subset\partial\overline{c_{\sigma}}} G_{\sigma'}\to \pi_{k}(G,e)
\]
which is unambiguous since $\pi_{k}(G,e)$ is abelian. Its kernel determines a relation that we will call \emph{boundary relation in homotopy}.
On the other hand, consider the group 
\[
\widetilde{H}_{\cC} =  \left(\prod_{\substack{c_{u},c_{v},c_{w}\in\cC_{n}\\\\ \overline{c_{u}}\cap\overline{c_{v}}\cap\overline{c_{w}} = \overline{c_{\sigma}}}}\left(\prod_{\overline{c_{\sigma'}}\subseteq\overline{c_{\sigma}}}H_{\sigma'}\right)\right)\times\left(\prod_{\substack{c_{v},c_{w}\in\cC_{n}\\\\\overline{c_{v}}\cap\overline{c_{w}}\neq\emptyset}}H_{vw}\right)
\]
where $H_{\sigma'} = \pi_{k}\left(V_{G},\mathbf{e}\right)$ if $c_{\sigma'}\in\cC_{k}$, and $H_{vw} = \pi_{n-1}(G,e)$. Consider the homomorphisms
\[
\mu_{vw}: H_{vw}\times H_{wv} \to \pi_{n-1}(G,e),\qquad ([g],[g']) \mapsto [gg']
\]
which, by a slight abuse of notation, will be regarded as defined on $\widetilde{H}_{\cC}\times\widetilde{H}_{\cC}$ after composing with a suitable projection map. In particular, $\ker(\mu_{vw}) = \ker(\mu_{wv})$. Similarly, for $c_{\sigma}\in\cC_{k+1}$, $k = 1,\dots, n-3$  define the product homomorphism
\[
\beta_{\sigma} : \prod_{c_{\sigma'}^{k}\subset\partial\overline{c_{\sigma}}}  \pi_{k}\left(V_{G},\textbf{e}\right) \to  \pi_{k}\left(V_{G},\textbf{e}\right)
\]
whose kernel defines an analogous boundary homotopy relation. Finally, we can consider a triadic action of $\mathrm{S}_{3}$ in $\widetilde{H}_{\cC}$, by acting in the obvious way in the first component, and as the identity in the second component. Such action determines an invariant subgroup $\widetilde{H}_{\cC}^{\mathrm{S}_{3}}\subset \widetilde{H}_{\cC}$. 

\begin{corollary}\label{cor:torsor} 
For any given LG field $\mathrm{PT}_{\Gamma}:\cP_{\Gamma}\to G$ in $\cM_{\Gamma}$, the set of ELG fields in $\mathrm{pr}_{\cC}^{-1}\left(\mathrm{PT}_{\Gamma}\right)$ is a principal homogeneous space for the group 
\[
G_{\cC} = \left(\bigcap_{\overline{c_{\sigma'}}\subseteq\overline{c_{\sigma}}\subset\overline{c_{\tau}}}\ker\left(\alpha_{\sigma'}\right)\right)\cap\left(\bigcap_{\overline{c_{\sigma''}}\subseteq\overline{c_{\sigma'}}\subset\overline{c_{\sigma}}\subset\overline{c_{\tau}}}\ker\left(\Theta^{\sigma''}_{\sigma'\sigma\tau}\right)\right)
\]
Similarly, the set of cores of ELG fields in $\mathrm{pr}_{\min}^{-1}\left(\mathrm{PT}_{\Gamma,\min}\right)$ for any $\mathrm{PT}_{\Gamma,\min}\in \cN_{\Gamma}$ is a principal homogeneous space for the subgroup 
\[
H_{\min} = \left(\bigcap_{\substack{\overline{c_{\sigma'}}\subseteq\overline{c_{\sigma}} = \overline{c_{u}}\cap\overline{c_{v}}\cap\overline{c_{w}}\\\\ c_{u},c_{v},c_{w}\in\cC_{n}}} \ker(\beta_{\sigma'})\right)\cap 
\left(\bigcap_{\substack{c_{v},c_{w}\in\cC_{n}\\\\ \overline{c_{v}}\cap\overline{c_{w}}\neq\emptyset}}\ker(\mu_{vw})\right)\cap \widetilde{H}_{\cC}^{\mathrm{S}_{3}}
\]
\end{corollary}

\begin{proof} 
That $G_{\cC}$ acts on the set of ELG fields compatible with a given lattice gauge field $\mathrm{PT}_{\Gamma}\in\cM_{\Gamma}$ is a consequence of lemma \ref{lemma:torsor}: the action of a given component on an extension class is given by pointwise left multiplication of representatives, the boundary relations in homotopy are introduced to ensure that the vanishing of the obstruction to the existence of extensions to the interior of a cell  $\partial\overline{c_{\sigma'}}$ on an ELG field is preserved under the action on $G_{\cC}$. Similarly, the multiplicative relations in homotopy are introduced to ensure that the compatibility conditions \eqref{eq:comp-cond-ext} of an ELG field are preserved under the action on $G_{\cC}$, a fact that is possible due to the abelian nature of the homotopy groups of $G$ and that the action is left and right simultaneously. Moreover, since for a given $k$-cell $c_{\sigma'}$, the set of extension classes to its interior given a fixed boundary homotopy data is a principal homogeneous space for $\pi_{k}(G,e)$, it readily follows that the stabilizer in $G_{\cC}$ of any ELG field is equal to the identity. 

The same applies to the group $H_{\min}$ acting on the set of cores compatible with a given groupoid homomorphism $\mathrm{PT}_{\Gamma,\min}:\cP_{\Gamma,\min}\to G$, recalling the defining relation \eqref{eq:clutching}, as a consequence of corollary \ref{cor:core}. The invariance of $H_{\min}$ under the triadic action ensures that the action of any of its elements in a core is also invariant under the corresponding triadic action. Moreover, the boundary relations in homotopy ensure that the vanishing of obstructions to existence of extensions are preserved under the action of $H_{\min}$ on a core. Together with the inversion relations determined by all $\ker(\mu_{vw})$, we ensure that conditions (a) and (b) in corollary \ref{cor:core} are preserved.
\end{proof} 

\begin{remark}
The dissection of the core of an ELG field in corollary \ref{cor:core}, or what is the same, the relative homotopy classes of clutching maps, is strictly dependent on the fixed values over 0-cells of the latter, i.e. the a priori choice of a collection of reference points 
\[
\displaystyle\bigsqcup_{c_{\sigma}\in\cC_{n-2}}\left.\left\{\mathbf{h}_{\sigma}\left(c_{\sigma''_{0}}\right) \,\right\vert \, c_{\sigma''_{0}}\in\cC_{0},\, c_{\sigma''_{0}}\subset \overline{c_{\sigma}}\right\}. 
\]
This illustrates the contrast between intrinsically \emph{local} objects (ELG fields), and the \emph{global} notion of cellular bundle data representing an equivalence class of principal $G$-bundles: if a different choice of reference points is made, there would still exist a corresponding relative homotopy class of clutching maps defining an equivalent principal $G$-bundle, but the new extension classes resulting from the dissection procedure cannot be compared to the former. The interest in getting a better understanding of the correspondence between ELG fields and equivalence classes of principal $G$-bundles motivates another notion of cellular equivalence of ELG fields, following from definition \ref{def:bundle data}. Namely, two ELG fields $\left\{\mathrm{PT}_{\cC}\right\}$, $\left\{\mathrm{PT}'_{\cC}\right\}$ are said to be \emph{cellularly equivalent} if the principal $G$-bundles they determine are equivalent. The covering maps  \eqref{eq:covering1} and \eqref{eq:covering2} provide an alternative characterization of cellular equivalence, since by definition the core of an ELG field is the minimal local homotopy data (relative to $\cC$) extending any given LG field that is necessary to reconstruct a principal $G$-bundle $P\to M$ up to equivalence. Such characterization is described in the following result.
\end{remark}

\begin{proposition}
Two ELG fields in $\cM_{\cC}$ yield equivalent principal $G$-bundles if and only if they lie in the same connected component in $\cM_{\cC}$.
\end{proposition}

\begin{proof}
The covering map \eqref{eq:covering2} was defined to correspond to the cellular equivalence in definition \ref{def:bundle data}, as the connected components of $\cN_{\cC}$ are in natural correspondence with the equivalence classes of principal $G$-bundles. Then, let us assume that two ELG fields project to the same connected component in $\cN_{\cC}$, i.e. any pair of representatives of their cores are cellularly equivalent. Using an isomorphism, we may assume that the bundles they determine are the same, equal to $P$. Consider any pair of gauge classes of smooth connections representing $\left\{\mathrm{PT}_{\cC}\right\}$ and $\left\{\mathrm{PT}'_{\cC}\right\}$. Since the space $\cA_{P}/\cG_{P}$ is connected, there is a  path connecting both gauge classes of connections. Such path determines a path in $\cM_{\cC}$ connecting $\left\{\mathrm{PT}_{\cC}\right\}$ and $\left\{\mathrm{PT}'_{\cC}\right\}$, and therefore, the inverse image under $\mathrm{pr}_{\textrm{core}}$ of a connected component in $\cN_{\cC}$ is a connected component in $\cM_{\cC}$.
This determines a bijection between the connected components of $\cM_{\cC}$ and $\cN_{\cC}$, and implies the result.
\end{proof}

\begin{corollary}\label{cor:fibration}
Cellular equivalence determines fibrations 
\[
\Phi : \cM_{\cC} \to \check{H}^{1}\left(M,G\right), \qquad \Psi : \cN_{\cC} \to \check{H}^{1}\left(M,G\right)
\]
such that $\Phi = \Psi\circ \mathrm{pr}_{\mathrm{core}}$, given by mapping the connected components in $\cM_{\cC}$ (resp. $\cN_{\cC}$) to their corresponding isomorphism classes of principal $G$-bundles. In particular, there exists subgroups $K_{\cC}\subset G_{\cC}$ and $K_{\min}\subset H_{\min}$ such that
\begin{equation}\label{eq:homo-space}
G_{\cC}/K_{\cC}  \cong H_{\min}/K_{\min} \cong \check{H}^{1}\left(M,G\right), 
\end{equation}
that are the subgroups of deck transformations in $\cM_{\cC}$ (resp. $\cN_{\cC}$) preserving any given connected component, and thus stabilizing any given equivalence class of principal $G$-bundles $\{P\}$.
\end{corollary}

Having understood the way a discretization of a gauge field leads to the recovery of the topology of a principal $G$-bundle, we can now proceed and describe the global geometric picture that gives rise to the space of ELG fields $\cM_{\cC}$. In general, $\cM_{\cC}$ decomposes into connected components
\[
\cM_{\cC} = \bigsqcup_{\{P\}\in\check{H}^{1}\left(M,G\right)} \cM_{P,\cC}.
\]
Let us consider any given principal $G$-bundle $P$, with associated space of smooth connections modulo gauge transformations $\cA_{P}/\cG_{P}$. The gauge group $\cG_{P}$ contains a normal subgroup $\cG_{P,*}$, the \emph{restricted gauge group} (cf. \cite{Ba91}), consisting of those gauge transformations whose value over the fibers of every base point in $(M,\cC)$ is the identity, and characterized by stabilizing any given smooth parallel transport map $\mathrm{PT}_{\cC}$. Therefore, there is a bijective correspondence between the space of smooth parallel transport maps on $(M,\cC)$ yielding a bundle isomorphic to $P$, and the quotient $\cA_{P}/\cG_{P,*}$. Moreover, there is a principal bundle 
\begin{equation}\label{eq:fibration-restricted}
\cA_{P}/\cG_{P,*} \to \cA_{P}/\cG_{P}, 
\end{equation}
whose structure group is the Lie group $\cG_{P,\cC} = \cG_{P}/\cG_{P,*}$.

\begin{definition}\label{def:micro}
 
For $\cE_{\cC}$ a set of fiber points in $P$ covering a set of based points $\cB_{\cC}$ in  $(M,\cC)$, we say that two  smooth connections over $\left(P,\cE_{\cC}\right)$ are $\cC$-\emph{equivalent}, or \emph{equivalent at scale} $\cC$, if their associated parallel transport maps project to the same point in $\cM_{P,\cC}$.
 By a \emph{microscopical deformation} of $A\in\cA_{P}$, relative to a cellular network in $(M,\cC)$, we mean a smooth path in $\cA_{P}$ preserving $\cC$-equivalence and starting at $A$.
\end{definition}

A fundamental by-product of theorem \ref{theo:PT} is the \emph{cellular homotopy fibration} 
\begin{equation}\label{eq:cell homotopy fibration}
\cA_{P}/\cG_{P,*}\to\cM_{P,\cC}, \quad\qquad \mathrm{PT}_{\cC} \mapsto \{\mathrm{PT}_{\cC}\}
\end{equation}
mapping a given parallel transport map to its induced ELG field. The essence of definition \ref{def:micro} is captured in a commutative diagram of principal $\cG_{P,\cC}$-bundles 

\begin{equation}\label{eq:comm-diag2}
\begin{CD}
\cA_{P}/\cG_{P,*} @>>> \cM_{P,\cC}\\
@VVV @VVV\\
\cA_{P}/\cG_{P} @>>> \cM_{P,\cC}/\cG_{P,\cC}
\end{CD}
\end{equation}
The spaces of geometric significance in the proposed ELG theory are the quotient spaces $\cM_{P,\cC}/\cG_{P,\cC}$, 
which play the role of \emph{finite dimensional analogs} of the spaces $\cA_{P}/\cG_{P}$, for every $\{P\}\in\check{H}^{1}\left(M,G\right)$.

An interesting consequence of corollaries \ref{cor:torsor} and \ref{cor:fibration} is the reconstruction of the space $\check{H}^{1}(M,G)$ as a homogeneous space \eqref{eq:homo-space}, in terms of homotopy data of a local nature. Understanding the space $\check{H}^{1}(M,G)$ is reduced to understanding the groups $G_{\cC}$ and $K_{\cC}$.  Such a hybrid (geometric/algebraic) structure deserves to be thoroughly studied, as it sheds new light into the spaces $\check{H}^{1}(M,G)$ on arbitrary manifolds, but the task seems to be far from trivial. We plan to return to it in a subsequent publication.

Altogether, we have constructed a series of fibrations that take the set of smooth parallel transport maps over $(M,\cC)$ as a starting point. They could be understood diagramatically as follows:\footnote{We have excluded the final topological fibration to the characteristic classes of a principal $G$-bundle in the present work, as we plan to present it in a separate publication. A similar fibration was already studied in \cite{PS86,PS90,PS93}.}

\vspace{3mm}

\centerline{
\begin{xy}
(36,70)*+{\left\{\parbox[d]{1.6in}{\centering Equivalence classes of smooth parallel transport maps $\mathrm{PT}_{\cC}$}\right\}}="a";
(100,50)*+{\left\{\parbox[d]{1.5in}{\centering Extended lattice gauge fields on $(M,\cC)$}\right\}}="b"; 
(100,20)*+{\quad\left\{\parbox[d]{1.6in}{\centering Homotopical cellular bundle data on $(M,\cC)$}\right\}}="c";
(36,20)*+{\left\{\parbox[d]{1.5in}{\centering Equivalence classes of principal $G$-bundles}\right\}}="d";
(70,0)*+{\left\{\parbox[d]{1.5in}{\centering Characteristic classes of principal $G$-bundles}\right\}}="e"; 
{\ar@{->}_{\text{\parbox[d]{0.8in}{\centering Induced bundle equivalence fibration}}\quad} "a";"d"}
{\ar@{->}_{\text{\parbox[d]{1.3in}{\centering Cellular equivalence $\qquad$ of clutching maps$\qquad$ }}} "a";"c"}
{\ar@{->}^{\qquad \text{\parbox[d]{0.8in}{\centering Cellular equivalence fibration}}}  "b";"c"}
{\,\ar@{<->}_{\text{\parbox[d]{0.8in}{\centering Homotopy equivalence}}\quad} "c";"d"}
{\ar@{->}^{\qquad \text{\parbox[d]{0.8in}{\centering Induced ELGT data}}} "a";"b"}
{\ar@{->}^{\qquad\text{\parbox[d]{0.8in}{\centering  Chern-Weil}}} "d";"e"}
{\ar@{-->} "c";"e"}
\end{xy}}


\begin{remark}
The previous spaces have several fundamental properties  in the study of quantum gauge theories on the lattice. As in the standard LGT, the spaces $\cM_{\Gamma}$ are connected Lie groups (that are compact if $G$ and $M$ are compact), which carry a measure inherited from the Haar measure on $G$. It follows that the manifolds $\cM_{\cC}$ and their connected components also inherit a measure, but they may fail to be compact. It is then important to determine if their connected components $\cM_{P,\cC}$ are always compact. The spaces of physical relevance are the quotients $\cM_{P,\cC}/\cG_{P,\cC}$, which although finite-dimensional, are not necessarily smooth. To understand the properties of such quotients, it is necessary to understand $\cM_{\cC}$ better. We plan to return to such questions in a separate publication. 
\end{remark}

Consider a submanifold $\iota: N\hookrightarrow M$, and a cotriangulation of $M$ which restricts to a cotriangulation of $N$, which we will also denote by $\cC$. Naively speaking, we are allowing $N$ to possess a boundary and corners. Let us denote the corresponding space of extended lattice gauge fields in $N$ by $\cN_{\cC}$. In particular, there are two special and fundamental cases:
\begin{enumerate}
\item $N = \overline{c_{\sigma}}$, for any $c_{\sigma}\in \cC$,
\item $N$ is a $(n-1)$-dimensional and such that $N = \partial M$, or such that $M = M_{1}\sqcup_{N} M_{2}$ is the identification of two manifolds with boundary $M_{1}$ and $M_{2}$, along a common boundary $N = M_{1}\cap M_{2}$.
\end{enumerate}

We conclude this section with the following corollary, whose proof is straightforward.
\begin{corollary}
Under the above assumptions, the following properties follow:
\begin{itemize}
\item[(i)] There is a projection $\mathrm{pr}: \cM_{\cC} \to \cN_{\cC}$,
\item[(ii)] If $M = M_{1}\sqcup_{N} M_{2}$ (where $M_{i}$ inherits the cell decomposition $\cC_{i}$, $i=1,2$), then $\cM_{\cC}$ is isomorphic to the submanifold of
\[
\cM_{\cC_{1}}\times \cM_{\cC_{2}}, 
\]
consisting of pairs of ELG fields $\left(\left\{\mathrm{PT}_{\cC_{1}}\right\},\left\{\mathrm{PT}_{\cC_{2}}\right\}\right)$ such that
\[
\mathrm{pr}_{1}\left(\left\{\mathrm{PT}_{\cC_{1}}\right\}\right) = \mathrm{pr}_{2}\left(\left\{\mathrm{PT}_{\cC_{2}}\right\}\right) \in \cN_{\cC}
\]
\item[(iii)] For any principal $G$-bundle $P\to M$, let $P_{N} = P|_{N}$, $P_{i} = P|_{M_{i}}$, and $\widetilde{\cG}_{P_{i},\cC_{i}}$, $i = 1,2$, be the subgroup of $\cG_{P_{i},\cC_{i}}$ consisting of classes of gauge transformations that are trivial along $P_{N}$. Then, $\cM_{P,\cC}/\cG_{P,\cC}$ is isomorphic to the quotient 
\[
\cM_{P_{N},\cC_{N}}/\cG_{P_{N},\cC_{N}}
\]
under diagonal action, where $\cM_{P_{N},\cC_{N}}$ is the submanifold of 
\[
\left(\cM_{P_{1},\cC_{1}}/\widetilde{\cG}_{P_{1},\cC_{1}}\right)\times \left(\cM_{P_{2},\cC_{2}}/\widetilde{\cG}_{P_{2},\cC_{2}}\right)
\]
consisting of pairs of orbits of ELG fields whose induced projections to $\cN_{P_{N},\cC_{N}}$ coincide.
\end{itemize}
\end{corollary}

\section{Pachner moves and the dependence on cellular decompositions}\label{sec:Pachner}

We now address the dependence of our constructions on the underlying cotriangulation $\cC$. A special feature of the category of cotriangulations on a manifold $M$ is the abundance of morphisms that allows us to systematically connect and compare any given pair of them, up to dual P.L. equivalence. Pachner \cite{Pach91} proved that any two smooth triangulations of a manifold are related by a sequence of so-called \emph{Pachner moves}. Let $X^{n+1}$ be an abstract $(n+1)$-simplex. We say that two different P.L. structures $\Delta:|K| \to M$ and $\Delta':|K'| \to M$ on $M$ differ by a Pachner move if there exists a pair of injective simplicial maps $\mu: L \to K$ and $\mu': L' \to K'$, where
\[
L = \bigcup_{l=0}^{k}X^{n}_{l}\subset \partial X^{n+1}, \qquad L' = \bigcup_{l=k+1}^{n+1}X^{n}_{l}\subset \partial X^{n+1}, 
\]
for some arbitrary labeling $X^{n}_{0},\dots,X^{n}_{n+1}$ of the $n$-simplices in $\partial X^{n+1}$ and $0\leq k \leq n$, such that 
\begin{itemize}
\item[(i)] the map $\left(\Delta'\right)^{-1}\circ \Delta: |K| \to |K'|$ is simplicial outside $\left|\mu\left(L\setminus\partial L\right)\right|$, 
\item[(ii)] $\Delta\left(\left|\mu(L)\right|\right) = \Delta'\left(\left|\mu'(L')\right|\right)$, 
\item[(iii)] $\Delta\left(\left|\mu(\sigma)\right|\right) = \Delta'\left(\left|\mu'(\sigma)\right|\right)$ for all $\sigma \in \partial L =\partial L'$.
\end{itemize}
Therefore, any two cotriangulations of $M$ are related by a sequence of dual Pachner moves. We will describe such transformations explicitly.

\begin{remark}\label{remark:Pachner simplex}
In the definition of a Pachner move, when $k+1$ different $n$-simplices of the simplicial complex $L$ meet at a common $(n-k)$-face $\sigma$, the $n-k+1$ different $n$-simplices of the simplicial complex $L'$ meet at a common $k$-face $\sigma'$, dual to $\sigma$ (when $k=0$, $L = \sigma$ and $\sigma'$ is its complementary vertex $v$ in $X^{n+1}$, and correspondingly for $k=n$). Therefore, when $\Delta$ and $\Delta'$ of $M$ differ by a Pachner move, the latter corresponds to the replacement of the $(n-k)$-face $\sigma$ by the $k$-face $\sigma'$ (unless $k=0$, where the Pachner move consists of a refinement of an $n$-simplex through the extra interior vertex $v$, and conversely for $k=n$), and there is a common refinement $\Delta''$ containing exactly an additional 0-simplex $v$ corresponding to the intersection of $\sigma$ and $\sigma'$. In turn, $\cC\setminus\{c_{\sigma}\} \cong \cC'\setminus\{c_{\sigma'}\}$, and the corresponding cotriangulation $\cC''$ differs from $\cC$ and $\cC'$ by an additional $n$-cell $c_{v}$ (figure \ref{figure4}). The $n$-cell closure $\overline{c_{v}}$ is topologically equivalent to the product $\overline{c_{\sigma}}\times\overline{c_{\sigma'}}$. 
\end{remark}

\begin{lemma}\label{lemma:dual Pachner deformation}
Any two cotriangulations $\cC$ and $\cC'$ of $M$, with triangulations $\Delta$ and $\Delta'$ related by a Pachner move as before, are degenerations of a smooth 1-parameter family of cotriangulations $\{\cC''_{t}\}_{t\in (0,1)}$, 
\[
\lim_{t \to 0^{+}}\cC''_{t} = \cC,\qquad \lim_{t\to 1^{-}}\cC''_{t} = \cC',
\]
with $\cC''_{t}$ dual to $\Delta''$, the common refinement of $\Delta$ and $\Delta'$ (remark \ref{remark:Pachner simplex}). 
\end{lemma}

\begin{proof}
Consider a choice of cellular decomposition $\cC''$ dual to the common refinement $\Delta''$, in such a way that $\cU_{u} = \cU_{\sigma} = \cU_{\sigma'}$, and $\cC''$ coincides with $\cC$ and $\cC'$ on $M\setminus\cU_{v}$. Make $\cC''$ correspond with $\cC''_{1/2}$. Using the fact that $\overline{c_{v}}\cong \overline{c_{\sigma}}\times\overline{c_{\sigma'}}$,
define a family  on $(0,1/2]$ by letting $\overline{c_{v}}^{t}$ degenerate to $\overline{c_{\sigma}}$, and  on $[1/2,1)$ by letting $\overline{c_{v}}$ degenerate to $\overline{c_{\sigma'}}$, in such a way that all cells outside $\cU_{v}$ remain constant, all cells neighboring $c^{t}_{v}$ transform into the corresponding cells neighboring $c_{\sigma}$, $c_{\sigma'}$, and the full family over $(0,1)$ is smooth (figure \ref{figure4}). In the case when $k=0$ (resp. $k=n$), one of the  two degenerations would not be present: topologically, the cellular decompositions for $t\in(0,1]$ (resp. $t\in[0,1)$) would be equivalent.
\end{proof}

\begin{remark}
The cotriangulation $\cC_{\Delta}$ on an $n$-sphere acquired by realizing it as the boundary of an $(n+1)$-simplex is also a triangulation. 
An important consequence is the following realization of the dual Pachner moves in $\left(S^{n},\cC_{\Delta}\right)$. Choose a $k$-cell $c_{\sigma}$ in $\cC_{\Delta}$. Then, there is a $(n-k)$-cell $c_{\sigma'}$ in $\cC_{\Delta}$, complementary to the interior of the star of $c_{\sigma}$. Upon the choice of a smooth hemisphere $H$ in $S^{n}$ separating $c_{\sigma}$ and $c_{\sigma'}$, we can identify each of the two components in $S^{n}\setminus H$ with the open sets $\cU_{\sigma}$ and $\cU_{\sigma'}$. More can be said about the topology of the cell degenerations: It is possible to foliate $S^{n}\setminus \{\overline{c_{\sigma}},\overline{c_{\sigma'}}\}$ as a collection of cell subcomplexes, parametrized by $(0,1)$, and each isomorphic to $\partial \overline{c_{v}}$. The foliation can be extended to a foliation of the closed disk $\overline{\DD^{n+1}}$, with the leave $\cL_{t}$ corresponding to $\overline{c_{v}^{t}}$. This is the case since there exists a diffeomorphism $(0,1)\times \overline{c_{v}} \cong \overline{\DD^{n+1}}\setminus\{\overline{c_{\sigma}},\overline{c_{\sigma'}}\}$.
\end{remark}

\begin{figure}[!ht]
  \caption{(A) Pachner move on a triangulated surface, with common refinement shown (center). (B) The corresponding degenerations on the dual cellular decomposition. 
  }\label{figure4}
  
   \vspace*{3mm}
  
  \centering
    \includegraphics[width=0.7\textwidth]{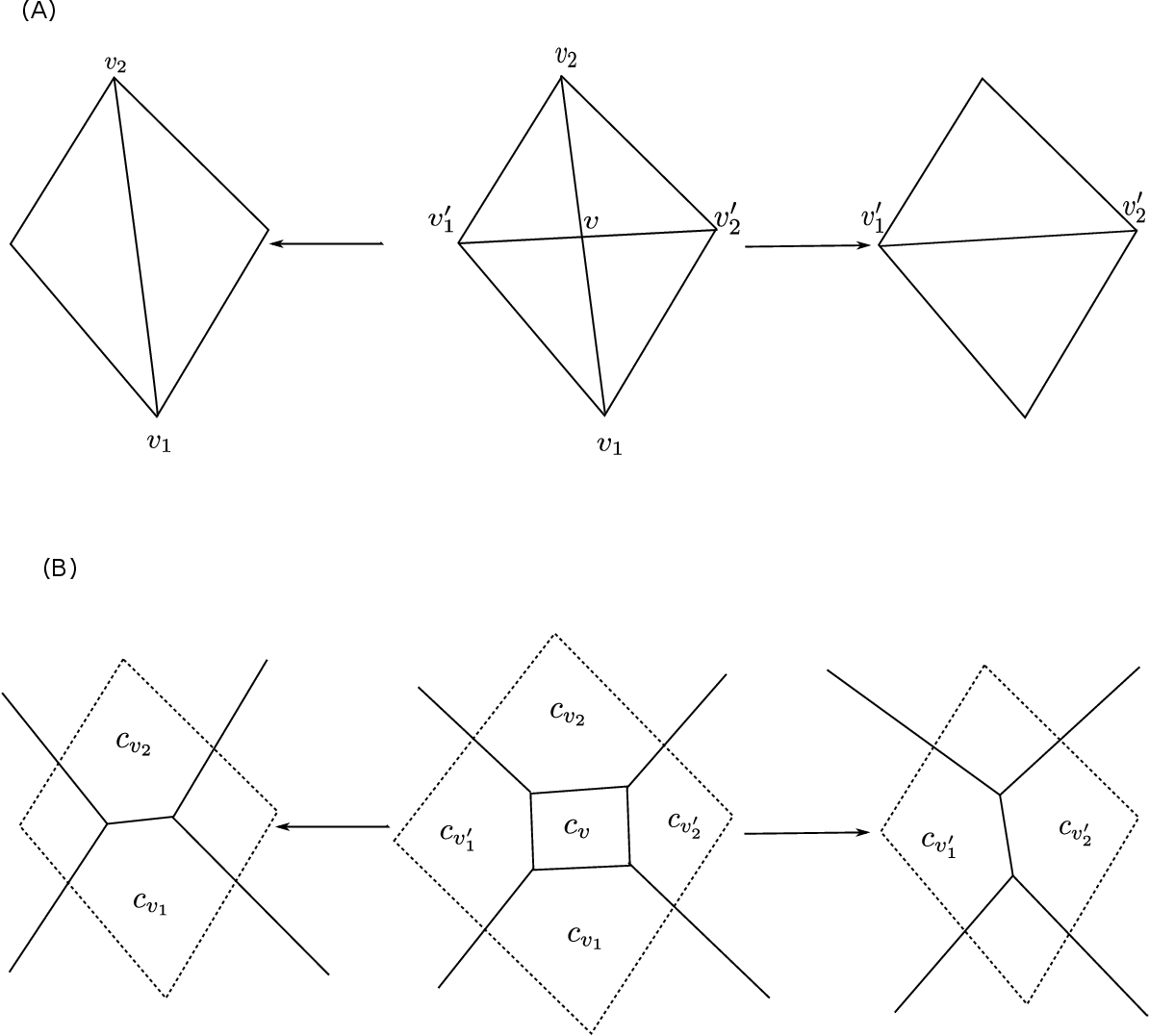}
\end{figure}

As an corollary of Pachner's theorem and lemma \ref{lemma:dual Pachner deformation}, we obtain the following result, useful to find sequences of transformations between cotriangulations.

\begin{corollary}
Any two smooth cotriangulations of $M$ are related by a sequence of smooth cotriangulations via deformations and contractions of cells.
\end{corollary}

\begin{remark}
In the same way that two cotriangulations $\cC$ and $\cC'$ in a mani\-fold $M$ related by a dual Pachner move can be thought of as degenerations of a 1-parameter family of cotriangulations $\{\cC''_{t}\}$,  a pair of choices of path groupoids $\cP_{\cD}$ and $\cP_{\cD'}$ adapted to $\cC$ and $\cC'$ may be understood as dege\-nerations of a 1-parameter family of path groupoids 
\[
\left\{\cP_{\cD''_{t}}\right\}_{t\in(0,1)}
\]
resulting from a 1-parameter family of cellular path groupoids interpolating $\cP_{\cD}$ and $\cP_{\cD'}$, relative to $\{\cC''_{t}\}_{t\in(0,1)}$. The characteristic feature of such degenerations follows from the isomorphism of cell complexes $\overline{c_{v}} = \overline{c_{\sigma}}\times\overline{c_{\sigma'}}$: for any pair of subcells $c_{\sigma_{0}}\subset\overline{c_{\sigma}}$ and $c_{\sigma'_{0}}\subset\overline{c_{\sigma'}}$, the subcell of $\overline{c_{v}}$ whose closure corresponds to $\overline{c_{\sigma}}\times \overline{c_{\sigma'}}$ would degenerate to $c_{\sigma_{0}}$ or $c_{\sigma'_{0}}$. Such degeneration implies the \emph{fusion} of the local path families supported over the closures of any two subcells in $\overline{c_{v}}$ collapsing to the same subcell in $\overline{c_{\sigma}}$ (or $\overline{c_{\sigma'}}$).  
\end{remark}

Consider a pair of cotriangulations $\cC$ and $\cC'$ related by a dual Pachner move as before, and coinciding over $M\setminus\cV_{\sigma} = M\setminus\cV_{\sigma'}$, where $\cV_{\sigma}$ (resp. $\cV_{\sigma'}$) denotes the open set in $M$ given as the interior of the star of $\overline{c_{\sigma}}$ --the union of all cells in $\cC$ (resp. $\cC'$) whose closures intersect $\overline{c_{\sigma}}$ (resp. $\overline{c_{\sigma'}}$). 
Let us moreover assume that $\cC$ and $\cC'$ are equipped with choices of ELG fields $\{\mathrm{PT_{\cC}}\}$ and $\{\mathrm{PT_{\cC'}}\}$ coinciding over all common cells  in $M\setminus \cV_{\sigma} = M\setminus\cV_{\sigma'}$. Consider a family $\cC''_{t}$ of cotriangulations degenerating to $\cC$ and $\cC'$ as before, and such that for all $t\in(0,1)$, the open set $\cV_{v}$ --the interior of the star of $\overline{c^{t}_{v}}$-- equals $\cV_{\sigma} = \cV_{\sigma'}$, and moreover, the restriction $\cC''_{t}|_{M\setminus\cV_{v}}$ coincides with the restrictions $\cC|_{M\setminus\cV_{\sigma}}$ and $\cC'|_{M\setminus\cV_{\sigma'}}$. Finally, equip such a family with a smooth family of adapted path groupoids $\left\{\cP_{\cD''_{t}}\right\}$, degenerating to $\cP_{\cD}$ and $\cP_{\cD'}$. 
The next definition plays the role of a \emph{generalized} cellular equivalence of ELG fields, suited to consider smooth 1-parameter families of cotriangulations arising from a dual Pachner move as before.

\begin{definition}
Two ELG fields $\{\mathrm{PT_{\cC}}\}$ and  $\{\mathrm{PT_{\cC'}}\}$ as before coinciding over all cells in $M\setminus \cV_{\sigma} = M\setminus\cV_{\sigma'}$ are called \emph{local relatives} if there exists a smooth 1-parameter family 
\[
\left\{\left\{\mathrm{PT_{\cC''_{t}}}\right\}\right\}_{t\in(0,1)}, 
\]
whose restriction to $\cC''_{t}|_{M\setminus\cV_{v}}$ coincides with the respective restrictions of $\{\mathrm{PT}_{\cC}\}$ and $\{\mathrm{PT}_{\cC'}\}$ to $\cC|_{M\setminus\cV_{\sigma}}$ and $\cC'|_{M\setminus\cV_{\sigma'}}$, and such that 
\[
\lim_{t\to 0^{+}} \left\{\mathrm{PT_{\cC''_{t}}}\right\} = \{\mathrm{PT_{\cC}}\},\qquad \lim_{t\to 1^{1}} \left\{\mathrm{PT_{\cC''_{t}}}\right\} = \{\mathrm{PT_{\cC'}}\}.
\]
\end{definition}

\begin{theorem}\label{theo:Pachner}
Let $\cC$ and $\cC'$ be two cotriangulations of $M$, related by a dual Pachner move, together with a pair of ELG fields $\{\mathrm{PT_{\cC}}\}$  and  $\{\mathrm{PT_{\cC'}}\}$ whose respective cores restricted to $\cC|_{M\setminus\cV_{\sigma}}$ and $\cC' |_{M\setminus\cV_{\sigma'}}$ coincide.  The principal $G$-bundles $P$ and $P'$ on $M$ induced by the cores $\{\mathrm{PT_{\cC}}\}_{\min}$ and $\{\mathrm{PT_{\cC'}}\}_{\min}$ are equivalent if and only if $\{\mathrm{PT_{\cC}}\}$ and $\{\mathrm{PT_{\cC'}}\}$ are local relatives.
\end{theorem}
\begin{proof}
The proof is straightforward. When the cores of the ELG fields are considered, or what is the same, the relative homotopy classes of collections of clutching maps, the notion of local relativity for a pair of ELG fields is a (local) specialization of the notion of cellular equivalence adapted to degenerating families of cell decompositions. 
The extra complications implicit in the notion of (global) cellular equivalence, present in general, disappear, since the homotopy equivalence relation that defines it has been confined to the $(n-1)$-cells in $\overline{c_{v}}$ and the interior of its star.
\end{proof}

\appendix

\section{Cotriangulations}\label{sec:triangle-dual}

Here we recall, for the sake of clarity and convenience, some standard notions on cell decompositions that we will require. We refer the reader to \cite{Lef42, Koz07} for further details.

Let $M$ be an $n$-dimensional smooth manifold, assumed to be connected, not necessarily compact, and with or without boundary. A \emph{smooth cell decomposition} of $M$ is defined inductively, as a collection $\cC= \{c_{\sigma}\subset M\}$ of pairwise disjoint subsets of $M$ (called \emph{cells}) whose union is equal to $M$, together with diffeomorphisms $\phi_{\sigma}^{k}:\DD^{k}\to M$ (where $\DD^{k}$ denotes the unit disk in $\RR^{k}$) 
such that $\overline{c_{\sigma}}\setminus c_{\sigma}$ 
is a union of cells in $\cC$ of dimension at most $k-1$. Two cell decompositions are equivalent if there is a diffeomorphism $\varphi:M\rightarrow M$ mapping one into the other.
$\cC_{k}$ will denote the collection of $k$-dimensional cells in $\cC$. The $l$th-skeleton of $\cC$ is defined as $\mathrm{Sk}_{l}(\cC)=\bigsqcup_{k=0}^{l}\cC_{k}$.

A \emph{flag} of length $m+1\leq n+1$ in $\cC$ of $M$ is a collection of nested cell closures $\overline{c_{\sigma_{0}}}\subset\overline{c_{\sigma_{1}}}\subset \dots\subset \overline{c_{\sigma_{m}}}$ in $M$. A flag is said to be \emph{gapless} if $\dim c_{\sigma_{k}}-\dim c_{\sigma_{k-1}} =1$ for every $1\leq k\leq m$. A flag is \emph{complete} if its length is equal to $n+1$, so that $\dim c_{\sigma_{k}}=k$. 

A barycentric subdivision $B(\cC)$ of $\cC$ is any cell decomposition of $M$ for which there is a 1--1 correspondence between $k$-cells in $B(\cC)$ and gapless flags of length $k+1$ in $\cC$, in such a way that if $b^{k}_{\sigma_{0}\dots\sigma_{k}}\in B(\cC)_{k}$ corresponds to $\{\overline{c_{\sigma_{0}}}\subset\dots\subset\overline{c_{\sigma_{k}}}\}\subset \cC$, then $b^{k}_{\sigma_{0}\dots\sigma_{k}}\subset c_{\sigma_{k}}$ (notice that $\dim c_{\sigma_{k}}\geq k$). In particular, there is exactly a 0-cell $b^{0}_{\sigma}\subset c_{\sigma}$ in $B(\cC)$ $\forall c_{\sigma}\in\cC$.

A smooth triangulation of $M$ is a homeomorphism $\Delta: |K| \to M$, where $K$ is an abstract simplicial complex and $|K|$ is its geometric realization, 
whose restriction to the interior of any simplex in $|K|$ is a diffeomorphism. We will denote by $\sigma, \tau,\dots$ the elements in $K$ (called the \emph{faces} of $K$), and by $v,w,\dots$ the vertices of it (in general, the maximal faces of an abstract simplicial complex $K$ are called its \emph{facets}, and $K$ is said to be \emph{pure} if all of its facets have the same dimension). Every triangulation has a canonical barycentric subdivision $B(\Delta)$. A classical result of J. H. C. Whitehead states that every smooth manifold $M$ admits a piecewise-linear (P.L.) structure, namely, a triangulation $\Delta$ for which the link of any simplex is a P.L. sphere, and that any two P.L. structures are related by a P.L. bijection. It is sufficient to assume the piecewise-smoothness of any such $\Delta$. 

\begin{definition}
A cell decomposition $\cC$ of $M$ is called a \emph{cotriangulation} if there exist an open cover $\mathfrak{U}=\{\cU_{v}\}_{c_{v}\in\cC_{n}}$  of $M$ with the following properties:
\begin{itemize}
\item[(i)] For every $\sigma\in N(\mathfrak{U})$ (the nerve of $\mathfrak{U}$, an abstract simplicial complex), the open set $\cU_{\sigma}= \cap_{v\subset\sigma}\cU_{v}$ is contractible, 

\item[(ii)] The geometric realization of $N(\mathfrak{U})$ is a P.L structure for $M$. In particular, $N(\mathfrak{U})$ is pure and $\cap_{i=0}^{n+1}\cU_{v_{i}} = \emptyset$ for all pairwise-different $c_{v_{0}},\dots,c_{v_{n+1}}\in\cC_{n}$,

\item[(iii)] There is a 1--1 correspondence between the $k$-simplices $\sigma\in N(\mathfrak{U})$ and the $(n-k)$-cells $c_{\sigma}\in\cC_{n-k}$, in such a way that $\overline{c_{\sigma}}\subset \cU_{\sigma}$.
\end{itemize}
\end{definition}

\begin{remark}
To avoid unnecessary pathologies, cell decompositions are assumed to have the following property:  for every $k$-cell $c_{\sigma}$ ($k\geq 1$), $\partial\overline{c_{\sigma}}$ is is a piecewise-smooth $(k-1)$-sphere (c.f. definition \ref{def:bundle data}). We call such cell decompositions \emph{spherical}. All cotriangulations are spherical (following from the P.L. structure), but the converse is not true: hypercubical cell decompositions of $\RR^{n}$ determined by a lattice $\Lambda\subset\RR^{n}$ are spherical but not cotriangulations, unless  $n=1$. Cotriangulations are generic among spherical cell decompositions (when $n = 1$, all cell decompositions are cotriangulations).
\end{remark}

\begin{figure}[!ht]
\caption{A cotriangulation $\cC$ on a surface.}\label{figure1}

\vspace*{3mm}

\centering
\includegraphics[width=0.5\textwidth]{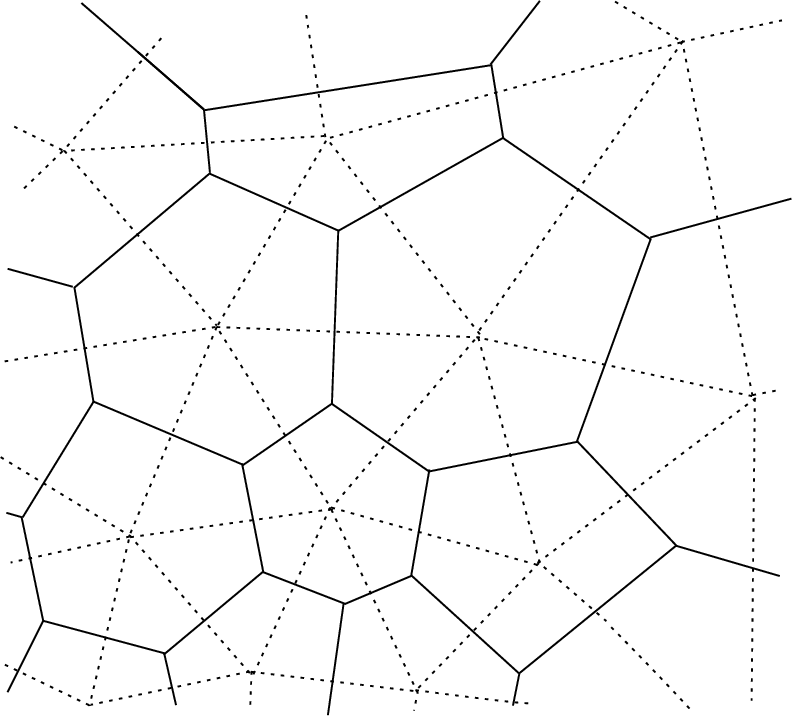}
\end{figure}

\begin{lemma}\label{lemma:stability}
Cotriangulations are stable under small deformations. Every spherical cell decomposition is a degeneration of a family of cotriangulations.  
\end{lemma}
\begin{proof}
A cell decomposition $\cC^{0}$ is unstable under small deformations if there exists a family $\cC^{\epsilon}$, $\epsilon\in\RR$ of cell decompositions which are equivalent for small $\epsilon\neq 0$ but inequivalent from $\cC^{0}$ (figure \ref{figure2}). In that case, there exists a family of $k$-cells $c_{\sigma}^{\epsilon}$ for some $0<k<n$ for which $\overline{c^{\epsilon}_{\sigma}}$ collapses to a 0-cell $c^{0}$ in $\cC^{0}$ as $\epsilon\to 0$. In particular all 0-cells in $\overline{c^{\epsilon}_{\sigma}}$ coalesce into $c^{0}$. A 0-cell in $\overline{c^{\epsilon}_{\sigma}}$ corresponds to an $n$-cell on the dual cell decompositions $\left(\cC^{\epsilon}\right)^{\vee}$. In the limit all of these $n$-cells unite into the single $n$-cell $(c^{0})^{\vee}$, and the number of $(n-1)$-cells in $\partial\overline{(c^{0})^{\vee}}$ would be larger that the number of $(n-1)$-cells in the boundary of any of the $n$-cells in $\left(\cC^{\epsilon}\right)^{\vee}$ that unite into $\partial\overline{(c^{0})^{\vee}}$. Since $\cC^{\vee}$ is a triangulation, there would be a family of $n$-cells in $\left(\cC^{\epsilon}\right)^{\vee}$ uniting into to the interior of an $n$-simplex in $\cC^{\vee}$, whose number of $(n-1)$-cells in its boundary is minimal (equal to n+1), a contradiction.  Hence $\cC$ cannot be a degeneration of a smooth family of equivalent cell decompositions $\cC^{\epsilon}$.

Given a spherical cell decomposition $\cC$ of $M$, it is always possible to refine $\cC^{\vee}$ into a triangulation $\Delta$ whose vertices coincide with the set  of $0$-cells in $\cC^{\vee}$, i.e. $\textrm{Sk}_{0}(\cC^{\vee}) = \textrm{Sk}_{0}(\Delta)$. Since the $n$-cells in $\Delta^{\vee}$ and $\cC$ are in  $1-1$ correspondence, it is possible to construct a family $\cC^{\epsilon}$ of equivalent cell decompositions in $M$ such that for sufficiently small $\epsilon\neq 0$, $\cC^{\epsilon}$ is equivalent to $\Delta^{\vee}$. The limit $\cC^{0}$ coincides with $\cC$ by construction.
\end{proof}

\begin{figure}[!ht]
  \caption{(A) 1-cell contraction in a cellular decomposition $\cC$ on a surface. (B) The effect over the dual decomposition $\cC^{\vee}$.}\label{figure2}
  \centering
    \includegraphics[width=0.5\textwidth]{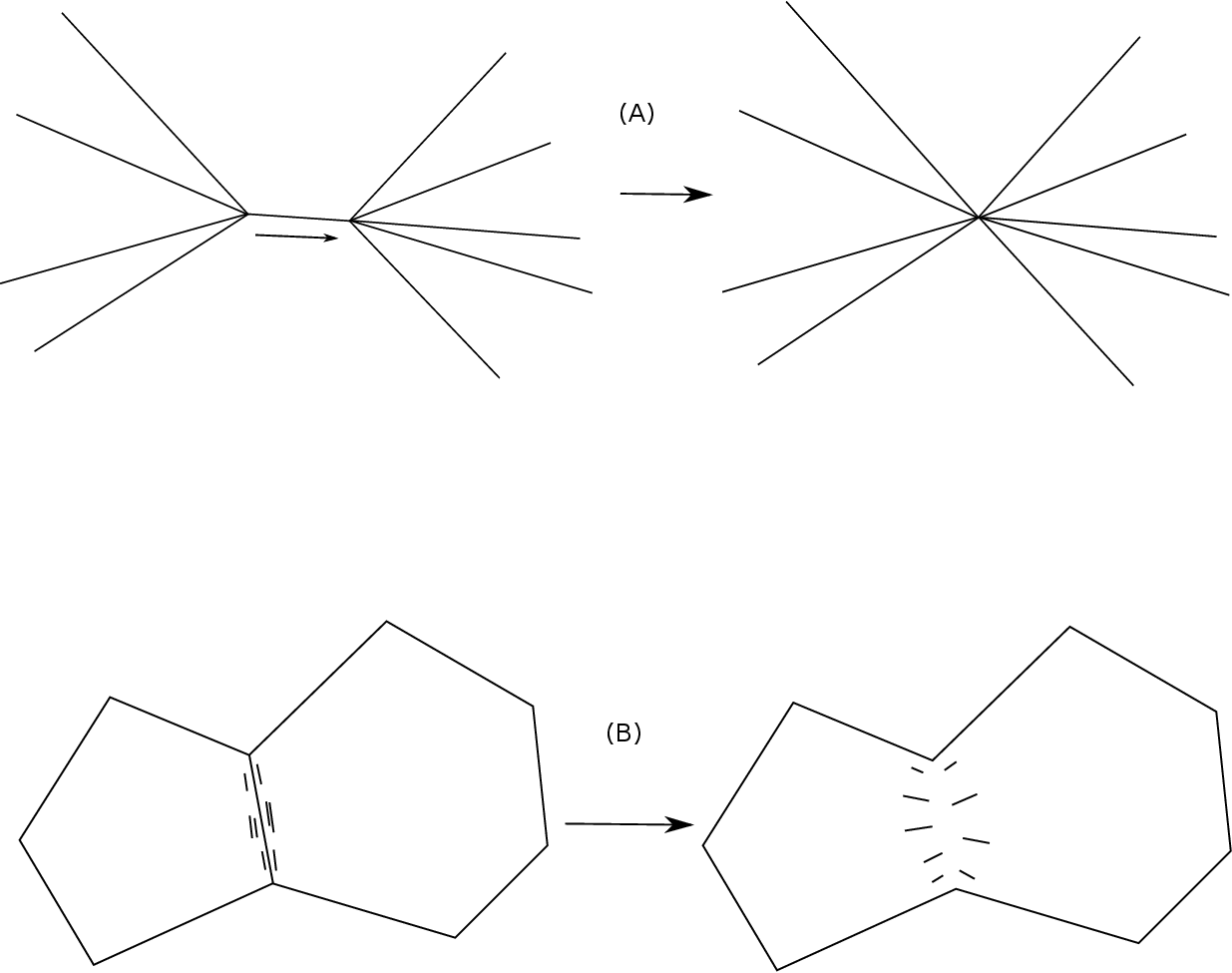}
\end{figure}

It is of paramount importance to relate the standard cell decompositions generated by a square lattice in $\RR^{n}$ to cotriangulations when the manifold $M$ is an Euclidean space or a quotient of it by  a lattice (i.e., a cylinder or a torus). The following lemma shows how the constructions and results of this work are also relevant and applicable in those standard cases.

\begin{lemma}
The cell decomposition $\cC^{0}$ of $\RR^{n}$ induced by the lattice $\Lambda^{0} = \ZZ^{n}$ is a degeneration of a family of cotriangulations $\cC^{\epsilon}$, $\epsilon\in(0,1)$.
\end{lemma} 
\begin{proof}
Setting a 1-1 correspondence between $n$-cells $c^{0}_{v}$ in $\cC^{0}$ and the 0-subcell $( i_{1}^{v},\dots, i_{n}^{v})$ in their closure, where for each $k=1,\dots,n$, 
\[
i_{k}^{v} = \min\{i_{k}\,:\, (i_{1},\dots,i_{n})\in\overline{c_{v}}\}, 
\]
we can identify the $n$-cells in $\cC^{0}$ with the elements in $\Lambda^{0}$. For any $\epsilon\in (0,1)$, consider the lattice $\Lambda^{\epsilon}$ in $\RR^{n}$, of $n$-tuples of real numbers of the form
\[
\left(i_{1},i_{2}+\epsilon i_{1},\dots,i_{n-1}+\epsilon i_{n-2},i_{n} + \epsilon i_{n-1}\right),\qquad i_{1},\dots,i_{n}\in\ZZ. 
\]
Then, in particular, $\Lambda^{0} = \ZZ^{n}$. Just as the lattice $\Lambda^{0}$ acts as the group of translation symmetries of $\cC^{0}$, the lattice $\Lambda^{\epsilon}$ will act on $\cC^{\epsilon}$ as its symmetry group of translations. The $n$-cell $c^{\epsilon}_{v}$, as well as its closure, is then defined by translating $c_{v}^{0}$ by the vector 
\[
\epsilon\left(0,i_{1}^{v},\dots,i_{n-1}^{v}\right),
\]
(figure \ref{figure6}).Although there is a 1-1 correspondence between $n$-cells in $\cC^{0}$ and $\cC^{\epsilon}$ by construction, new $k$-cells, $k=0,\dots,n-1$ are created in $\cC^{\epsilon}$ by subdivision of the boundary cells of each $n$-cell $c^{\epsilon}_{v}$. $\cC^{\epsilon}$ is a cotriangulation for every $\epsilon\in(0,1)$. Notice that over $\cC^{0}$, each 0-cell is the common intersection of $2n$ 1-cells. The $\epsilon$-shifts introduced to define $\cC^{\epsilon}$ separate $n-1$ of these 1-cells, leaving exactly $n+1$ 1-cells on each old and new 0-cell. The argument with the higher dimensional cells is similar, giving the combinatorial properties that determine a cotriangulation.
\end{proof}

\begin{figure}[!ht]
\caption{The cell decomposition $\cC^{\epsilon}$ in $\RR^{2}$.}\label{figure6}

\vspace*{3mm}

\centering
\includegraphics[width=0.4\textwidth]{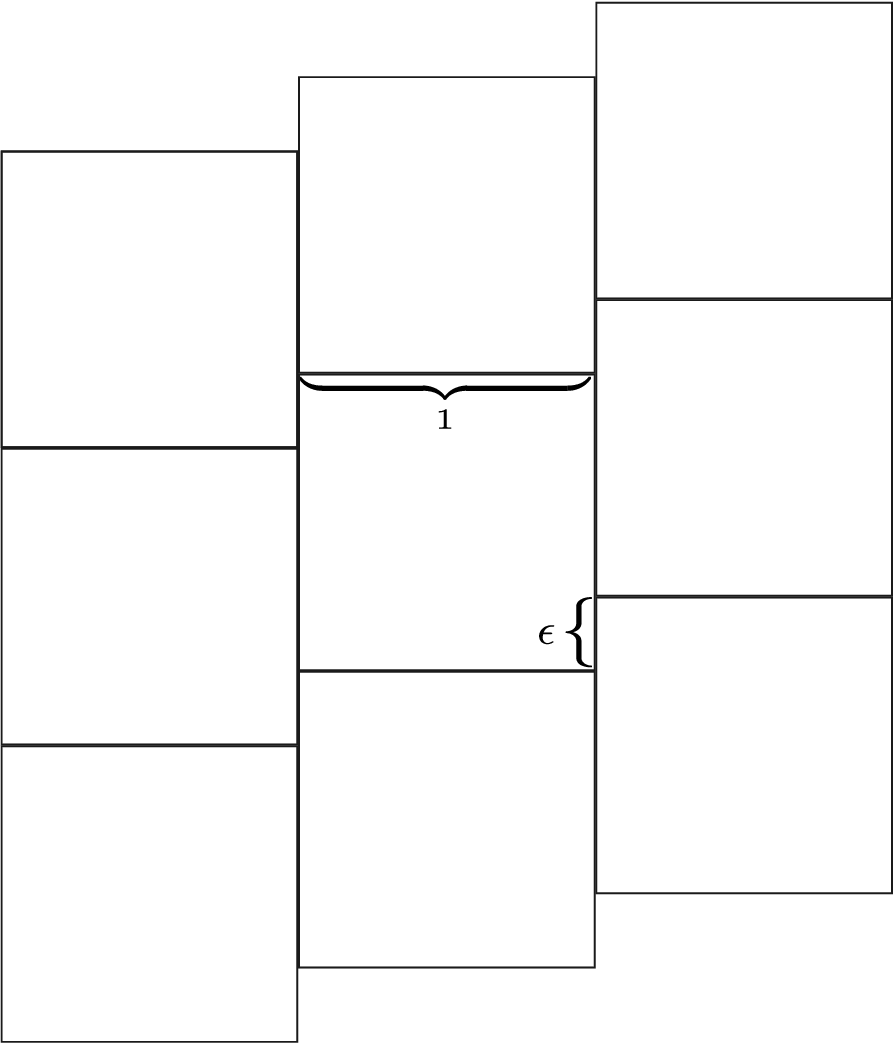}
\end{figure}

\section{Proof of the reconstruction theorem}\label{sec:parallel transport}

The idea to prove theorem \ref{theo:PT} (theorem 2 in \cite{Zapata12}) is to construct a concrete principal $G$-bundle $P$ from a smooth parallel transport map $\mathrm{PT}_{\cC}$. Then a horizontal lift property implies the existence of a gauge equivalence class of smooth connections in $P$, under the action of the group of bundle automorphisms covering the identity map in $M$, and acting as the identity over the fibers $\{\pi^{-1}(p_{\sigma})\}_{c_{\sigma}\in\cC}$, the \emph{restricted gauge transformations}. The group $\cG_{P,*} $ of restricted gauge transformations is a normal subgroup of the full gauge group $\cG_{P}$. 
Before proving theorem \ref{theo:PT}, we require to introduce the following result.

\begin{proposition}\label{prop:equivalence g}
(i) If $\mathrm{PT}_{\cC}$ and $\mathrm{PT}'_{\cC}$ are equivalent, then there is a set $\{g_{\sigma}\in G\}_{c_{\sigma}\in\cC}$ such that the glueing maps $\{g_{\sigma' \sigma}\}_{\overline{c_{\sigma}}\cap\overline{c_{\sigma'}}\neq \emptyset}$ and $\{g'_{\sigma' \sigma}\}_{\overline{c_{\sigma}}\cap\overline{c_{\sigma'}}\neq \emptyset}$ induced from a choice of cellular path groupoid $\cP_{\cD}$ are related as
\[
g'_{\sigma' \sigma}(x)=g_{\sigma'}\cdot g_{\sigma' \sigma}(x)\cdot g_{\sigma}^{-1}.
\]
(ii) If the choice of  $\cP_{\cD}$ is changed, then there is a set of smooth maps $\{g_{\sigma}(x):\overline{c_{\sigma}}\to G\}_{c_{\sigma}\in\cC}$, such that the maps $\{g_{\sigma' \sigma}\}_{\overline{c_{\sigma}}\cap\overline{c_{\sigma'}}\neq \emptyset}$ induced from $\mathrm{PT}_{\cC}$ transform as
\[
g_{\sigma' \sigma}(x)\mapsto g_{\sigma'}(x) g_{\sigma' \sigma}(x) g_{\sigma}(x)^{-1}.
\]
\end{proposition}

\begin{proof}
(i) is a straightforward consequence of definition \ref{def:PT}. Namely, the value of $g_{\sigma' \sigma}(x)$ (and similarly for the value $g'_{\sigma' \sigma}(x)$), for an arbitrary $x\in\overline{c_{\sigma}}\cap\overline{c_{\sigma'}}$, is equal to $\mathrm{PT}_{\cC}\left(\left[\gamma^{x}_{\sigma' \sigma}\right]\right)$. (ii) follows as a consequence of lemma \ref{lemma:independence}. The proposition applies in particular to the clutching maps $h_{v w} = g_{v w}$.
\end{proof}

\begin{proof}[\textbf{Proof of theorem} \ref{theo:PT}]
Given a smooth parallel transport map $\mathrm{PT}_{\cC}$, we can define a collection of principal $G$-bundles $\pi_{\sigma}:P_{\sigma}\to \overline{c_{\sigma}}$ for every $c_{\sigma}\in\cC$ in the following way. Let $\cP^{s}_{\sigma} = \bigcup_{p_{\sigma'}\in\cB_{\sigma}} \cP(\overline{c_{\sigma}},p_{\sigma'})$ and define

\begin{equation}\label{eq:quotient}
P_{\sigma}=\cP^{s}_{\sigma}\times G /_{\sim\mathrm{PT}} 
\end{equation}
where two pairs $([\gamma],g)$ and $([\gamma'],g')$ are related if $t(\gamma)=t(\gamma')$ and 
\[
g' = \textrm{PT}_{\sigma}\left(\left[\gamma'\right]^{-1}\cdot [\gamma]\right) g.
\]
The projection of a class $[\gamma,g]$ onto $M$ is simply defined as $[\gamma,g] \mapsto t(\gamma)$, determining a map $\pi_{\sigma}: P_{\sigma}\to \overline{c_{\sigma}}$. The global right $G$-action on $P_{\sigma}$ is defined as $[\gamma,g']\cdot g = [\gamma,g'g]$.
For every $p_{\sigma'}\in \cB_{\sigma}$, there is a special point $b_{\sigma'}\in\pi_{\sigma}^{-1}(p_{\sigma'})$, determined by considering the classes  $[p_{\sigma'},e]$, where by a slight abuse of notation, $[p_{\sigma'}]$ represents the class of the constant path at $p_{\sigma'}$. Consequently, there is an identification of the fiber $\pi_{\sigma}^{-1}(p_{\sigma'})$ with $G$. 
A smooth and global \emph{cellular trivialization} can be given for each principal $G$-bundle $\pi_{\sigma}:P_{\sigma}\to \overline{c_{\sigma}}$. This can be seen by considering the smooth family of paths $\cF^{s}_{\sigma}\subset\cP(\overline{c_{\sigma}}, p_{\sigma})$ 
induced by a diffeomorphism $\psi_{\sigma}:\overline{\DD^{k}}\to \overline{c_{\sigma}}$ such that $\psi_{\sigma}(0) = p_{\sigma}$, as the collection consisting of the images of linear segments in $\overline{\DD^{k}}$ with source at 0. This way, we get a bijection
\[
\Psi_{\sigma}:\overline{c_{\sigma}}\times G \to P_{\sigma}, \qquad (x,g)\mapsto \left[\gamma^{x},g\right],	
\]
which defines a smooth structure of manifold with corners on $P_{\sigma}$, and a trivialization as a smooth principal $G$-bundle over $\overline{c_{\sigma}}$.
Since whenever $c_{\sigma}\subset \overline{c_{\tau}}$, we have that $\cP^{s}_{\sigma}\subset\cP^{s}_{\tau}$, we can construct a smooth bijection $P_{\sigma}\mapsto P_{\tau}|_{\overline{c_{\sigma}}}$ by restriction. Therefore, the collection of principal $G$-bundles $\{P_{\sigma}\}_{c_{\sigma}\in\cC}$ can be glued into a single smooth bundle $\pi:P \to M$.\footnote{The bundle $P$ could also be constructed as a quotient similar to \eqref{eq:quotient}, in terms of the full path space $\cP^{s}_{\cC}$ of all intimacy equivalence classes of piecewise-smooth paths with source an arbitrary $p_{\sigma}$ in $M$, and projection $\pi$ defined in a similar way.} As before, such a bundle would come with a preferred set of points $\{b_{\sigma}\in\pi^{-1}(p_{\sigma})\}_{c_{\sigma}\in\cC}$, which we will denote by $\cE_{\cC}$. Thus, $\pi\left(\cE_{\cC}\right) = \cB_{\cC}$ and any other choice of complete cellular path families $\cF$ would determine an equivalent bundle.

A \emph{horizontal lift} for every path  $\gamma:[0,1]\to M$ and a choice of initial condition $[\gamma' ,g]\in \pi^{-1}(\gamma(0))$ can be constructed from the smooth parallel transport map $\mathrm{PT}_{\cC}$. Namely, horizontal lifts can first be defined  if we take any path $\gamma:[0,1]\to M$ and an initial condition $[\gamma', g]$, $[\gamma']\in\cP^{s}_{\cC}$. Then by the independence under reparametrization, the path $\lambda(\gamma) : [0,1] \to P_{\sigma}$ with source $s(\lambda(\gamma))=[\gamma',g]$ is defined as 
\[
[\lambda(\gamma)](\epsilon)=\left[\left(\gamma \cdot \gamma'\right) |_{[0,(\epsilon+1)/2]},g\right],\quad \epsilon \in [0,1]
\]
satisfying $\pi \circ \lambda(\gamma)(\epsilon) = \gamma (\epsilon)$. Horizontal lifts may also be constructed as an iteration of local lifts 
and subsequent initial conditions, following lemma \ref{lemma:path factorization}.
Therefore, the construction of a gauge orbit of connections under the action of $\cG_{P,*}$ follows. The construction depends on the choice of base points $\cE_{\cC}$ along the fibers over the base points $\cB_{\cC}$, and determines a smooth connection in $P$, up to the action of $\cG_{P,\*}$. 
\end{proof}

\begin{remark}\label{remark:skeletal-normalization}
The bundle $P$ can also be constructed from local trivializations and clutching maps.
Consider a collection of smooth path families $\{\cF^{s}_{\sigma}\}_{c_{\sigma}\in\cC}$, $\cF^{s}_{\sigma} = \{\gamma^{x}_{\sigma}\,:\, x\in\overline{c_{\sigma}}\}$ as before. Whenever $\tau \subset \sigma$, the identity
\[
[\gamma_{\tau}^{x},g]=[\gamma^{x}_{\sigma},g']
\]
for any $x\in\overline{c_{\sigma}}\subset\overline{c_{\tau}}$, together with the corresponding trivializations, allows us to express $g'=g_{\sigma\tau}(x)g$ for some $g_{\sigma\tau}(x)\in G$. 
If $\{\cF_{\sigma}^{s}\}$ is changed, $g_{\sigma\tau}(x)$ would transform as 
\[
g_{\sigma\tau}(x)\mapsto g_{\sigma}(x)g_{\sigma\tau}(x) g_{\tau}(x)^{-1}, 
\]
for some well-defined smooth functions $g_{\sigma}:\overline{c_{\sigma}}\to G$ and $\left(g_{\tau}:\overline{c_{\tau}}\to G\right) |_{\overline{c_{\sigma}}}$, as it follows from proposition \ref{prop:equivalence g}. Only the homotopy type of the glueing maps $\{g_{\sigma' \sigma}\}$, is relevant to determine an equivalence class of principal $G$-bundles on $M$, and ultimately, the only data that is strictly necessary to reconstruct a bundle $P$ are the clutching maps \eqref{eq:clutching}, 
that determine a choice of \emph{cellular bundle data}, characterizing a principal $G$-bundle with trivializations over $\cC$, up to equivalence, as it follows from theorem  \ref{theo:equivalence}.\footnote{It is possible to construct a system of transition functions for $P$ if each $P_{\sigma}$ is instead constructed over an open set $\cU_{\sigma}$, by extending the spaces $\cP^{s}_{\sigma}$ to consist of paths belonging to $\cU_{\sigma}$ (and not only to $\overline{c_{\sigma}}$), and repeating the previous constructions verbatim.} 
\end{remark}


\section{Proof of the equivalence theorem for bundle data}\label{app:cellular bundle data}

Given an open cover $\mathfrak{U}=\{\cU_{v}\}$, a \v{C}ech 1-cocycle is a collection of smooth maps $g_{v w}:\cU_{v w}= \cU_{v}\cap\cU_{w}\rightarrow G$ satisfying
\begin{equation}\label{eq:cocycle1}
g_{w v}=g_{v w}^{-1}\quad\quad \text{on}\quad\cU_{v w}, \qquad\qquad g_{uv}g_{v w} = g_{uw} \quad\quad \text{on}\quad \cU_{uvw},
\end{equation}
and two cocycles $\{g_{v w}\}$ and $\{g'_{v w}\}$ are cohomologous if there exist a collection of smooth maps $\{ g_{v}:\cU_{v}\to G \}_{c_{v}\in\cC_{n}}$ (\emph{local gauge transformations}) such that 
\begin{equation}\label{eq:cocycle2}
g'_{v w}=g_{v}g_{v w}g^{-1}_{w} \quad\quad\text{on}\quad \cU_{v w}.
\end{equation}
These are statements about the 2-skeleton of the nerve $N(\mathfrak{U})$. It that sense it is convenient to work with cotriangulations and adapted open covers. A cover $\mathfrak{U}$ is \emph{adapted to $\cC$} if it is good, $N(\mathfrak{U}) \cong \cC^{\vee}$, and for every $c_{\sigma}\in\cC$, the open set $\cU_{\sigma}\in\mathfrak{U}$ satisfies $\overline{c_{\sigma}}\subset \cU_{\sigma}$. 

\begin{theorem}\label{theo:equivalence}
Let $M$ be an $n$-manifold, $n\geq 2$, together with a cotriangulation $\cC$ and an adapted cover $\mathfrak{U}$. 
There is a bijective correspondence
\[
 \left\{\parbox[d]{1.4in}{\centering Cellular bundle data in $M$, relative to $\cC$}\right\} \longleftrightarrow \left\{\parbox[d]{1.6in}{\centering Cohomology classes of \v{C}ech 1-cocycles on $\mathfrak{U}$}\right\}
\]
\end{theorem}
\begin{proof}
Consider a class of \v{C}ech 1-cocycles $\{g_{v w}\}$ as above. For any given representative and any $c_{\tau}\in\cC_{n-1}$ such that $\overline{c_v}\cap\overline{c_w}=\overline{c_{\tau}}$, the restrictions $h_{v w} := g_{v w}|_{\overline{c_{\tau}}}$ determine a representative of a class of cellular bundle data. This is so since every $\cU_{v}$ is contractible by definition, and any local gauge transformation $g_{v}:\cU_{v}\to G$ is smoothly homotopic to the identity. Therefore, any equivalent cocycle $\{g'_{v w}\}$ would be homotopic to $\{g_{v w}\}$ trough a homotopy of cocycles, inducing a cellular equivalence of maps $\{h_{v w}(t)\}$, $t\in[0,1]$. 

Conversely, consider an arbitrary choice of cellular bundle data over $\cC$, and for each $c_{v}, c_{w}\in\cC_{n}$ such that $\overline{c_v}\cap\overline{c_w}=\overline{c_{\tau}}$ with $c_{\tau}\in\cC_{n-1}$, fix data representatives in the form of maps $h_{v w}:\overline{c_{\tau}}\to G$ such that for any $(n-2)$-cell closure $\overline{c_{\sigma}}$, the restriction of the induced triples of maps to the diagonal $\Delta(\overline{c_{\sigma}}\times\overline{c_{\sigma}}\times\overline{c_{\sigma}})$ lies in $V_{G}$. We will construct maps $g_{v w}:\cU_{v w}\to G$ 
satisfying the cocycle conditions \eqref{eq:cocycle1} for any triple $\{c_{v_{1}},c_{v_{2}},c_{v_{3}}\}\in\cC_{n}$ as in \eqref{eq:cyclic triples}. The construction will be done in two steps.

First let us consider for every $c_{\tau}\in\cC_{n-1}$, the intersections $\cU_{\tau}\cap \textrm{Sk}_{n-1}(\cC)$ (figure \ref{figure3}), and their subsets 
\[
\cZ_{\tau}=\cU_{\tau}\cap \left(\bigcup_{\tau'\in\cC^{\tau}_{n-1}}\overline{c_{\tau'}}\right), 
\]
where $\cC_{n-1}^{\tau}=\{c_{\tau'}\in\cC_{n-1}\, \vert \, \overline{c_{\tau}}\cap\overline{c_{\tau'}}=\overline{c_{\sigma}},\; c_{\sigma}\in\cC_{n-2}\}$. Each $\cZ_{\tau}$ is topologically a cylinder for $\partial\overline{c_{\tau}}$ without boundary. We will extend each $h_{v w}|_{\partial\overline{c_{\tau}}}$ to the whole $\cZ_{\tau}$ in such a what that the cocycle conditions are still satisfied. Consider any triple $\{c_{v_{1}}, c_{v_{2}}, c_{v_{3}}\}$ as in \eqref{eq:cyclic triples}. 
For each $i=1,2,3$ $j,k\neq i$,  $(ijk)\sim (123)$, consider the intersection 
\[
\cI_{i}=\mathrm{pr}_{i}^{-1}\left(h_{v_{j}v_{k}}(\overline{c_{\tau_{i}}} \cap \cU_{\sigma})\right)\cap V.
\]
For $i=1,2,3$,  and $c_{\sigma}$ as above, consider any collection of piecewise-smooth maps $H^{i}_{\sigma}:\left(\overline{c_{\tau_{1}}}\cup\overline{c_{\tau_{2}}}\cup\overline{c_{\tau_{3}}}\right)\cap \cU_{\sigma}\to G\times G\times G$,  satisfying that
\begin{itemize}
\item[(i)] $H^{i}_{\sigma}\left((\overline{c_{\tau_{1}}}\cup\overline{c_{\tau_{2}}}\cup\overline{c_{\tau_{3}}})\cap \cU_{\sigma}\right)\subset \cI_{i}$, 
\item[(ii)] $H^{i}_{\sigma} |_{\overline{c_{\sigma}}}=(h_{v_{1}v_{2}}, h_{v_{2}v_{3}}, h_{v_{3}v_{1}})|_{\overline{c_{\sigma}}}$,
\item[(iii)] If for $j\neq i$ there is $c_{\tau'_{j}}$ such that $\mathrm{int}\left(\overline{c_{\tau_{j}}}\cap\overline{c_{\tau'_{j}}}\right)\subset\cZ_{\tau_{i}}\setminus\overline{c_{\tau_{i}}}$,\footnote{Here, $\mathrm{int}\left(\overline{c_{\tau_{j}}}\cap\overline{c_{\tau'_{j}}}\right)$ denotes the $(k-2)$-cell whose closure is $\overline{c_{\tau_{j}}}\cap\overline{c_{\tau'_{j}}}$.} then 
\[
H^{i}_{\sigma} |_{\mathrm{int}\left(\overline{c_{\tau_{j}}}\cap\overline{c_{\tau'_{j}}}\right)} = H^{i}_{\sigma'} |_{\mathrm{int}\left(\overline{c_{\tau_{j}}}\cap\overline{c_{\tau'_{j}}}\right)}.
\]
\end{itemize}
Such collections necessarily exist, since each of the sets $(\overline{c_{\tau_{1}}}\cup\overline{c_{\tau_{2}}}\cup\overline{c_{\tau_{3}}})\cap \cU_{\sigma}$ deformation retracts to $\overline{c_{\sigma}}$. The set of maps $\{H^{i}_{\sigma}\}$ is parametrized by the elements $c_{\sigma}\in\cC_{n-2}$ according to \eqref{eq:cyclic triples}. 

Altogether, the functions $\{H^{i}_{\sigma}\}_{c_{\sigma\in\cC_{n-2}}}$ determine extensions of each $h_{v_{j}v_{k}}$ to $\cZ_{\tau_{i}}$, $(ijk)\sim (123)$. The extension to a subset $\cU_{\tau_{i}}\cap \overline{c_{\tau_{j}}}$ is defined by means of the function $H^{j}_{\sigma}$. Conditions (ii) and (iii) above ensure that the overall process gives a well-defined extension of $h_{v_{j}v_{k}}$ over $\cZ_{\tau_{i}}$, and condition (i) ensures that the new functions would satisfy the cocycle condition \eqref{eq:cocycle1}.
Denote such extension by $h_{v_{j}v_{k}}$ again, now as a function on $\cZ_{\tau_{i}}\cup c_{\tau_{i}}$.


\begin{figure}[!ht]
  \caption{The intersection of $\cU_{\tau}$ with the $(n-1)$-skeleton of $\cC$.}\label{figure3}
  
  \vspace*{3mm}
  
  \centering
    \includegraphics[width=0.34\textwidth]{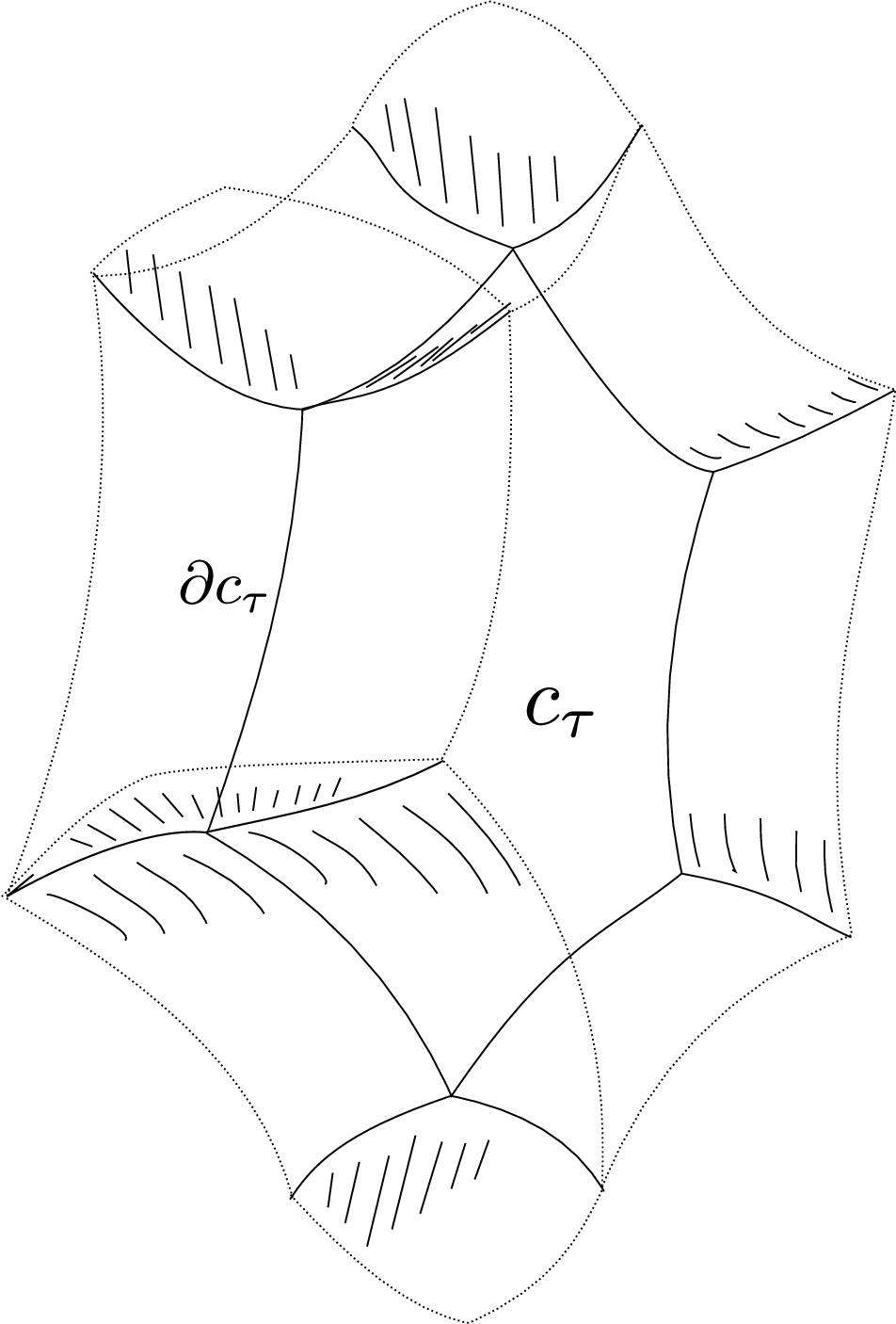}
\end{figure}

The second step of the construction  is to extend each $h_{v_{j}v_{k}}$ to the whole $\cU_{\tau_{i}}$, and relies in the following observation. For any triple $\{c_{v_{1}}, c_{v_{2}}, c_{v_{3}}\}$ as in \eqref{eq:cyclic triples}, the sets
\[
\cB_{\tau_{i}}=\cU_{\tau_{i}}\cap \left(c_{\tau_{i}}\cup c_{v_{j}}\cup c_{v_{k}}\right), \qquad i\neq j,k,\quad j\neq k, 
\]
which are open cylinders for $c_{\tau_{i}}$  (whose boundary points contain $\cZ_{\tau_{i}}$), satisfy $\cB_{\tau_{1}}\cap\cB_{\tau_{2}}\cap\cB_{\tau_{3}}=\emptyset$ (figure \ref{figure3}). If we consider any piecewise-smooth extension of $h_{v_{j}v_{k}}$ to $\cB_{\tau_{i}}$ for each $c_{\tau_{i}}$, then there is a unique way to extend such maps to the sets $\cU_{\tau_{i}}\setminus\overline{\cB_{\tau_{i}}}$ by forcing the cocycle condition on all triple intersections $\cU_{\sigma}$, obtaining maps defined over $\cU_{\tau_{i}}$.  
Hence, we can determine a collection $\{g_{v w}:\cU_{v w}\to G\}$, defining a cocycle.  This set of transition functions depends a priori on the choice of representatives $h_{v w}:\overline{c_{v}}\cap\overline{c_{w}}\to G$ for each $c_{v}, c_{w}\in\cC_{n}$ such that $\overline{c_{v}}\cap\overline{c_{w}} = \overline{c_{\tau}}$ with $c_{\tau}\in\cC_{n-1}$,  and the choice of extensions of such $h_{v w}$ to the sets $\cZ_{\tau}\cup\cB_{\tau}$.

To conclude, we must show that another set of representatives for the choice of cellular bundle data or another choice of extensions of the maps $h_{vw}$ determine a cocycle which is cohomologous to the previous one. But any choice of extensions of any pair of representatives $h_{v w}$ and $h'_{v w}$ of a given class $[h_{v w}]$ would be homotopic  through cellularly-smooth functions satisfying the cocycle conditions, hence the resulting homotopies $g_{v w}(t)$ between each $g_{v w}$ and $g'_{ v w}$ would define transition functions for every $t\in [0,1]$, equivalent to a principal $G$-bundle over $M\times[0,1]$. It follows from the homotopy invariance properties of principal bundles \cite{MS74}, that  the bundles on $M$ resulting by restriction to the boundary of $M\times[0,1]$ are isomorphic.
\end{proof}

\noindent \textbf{Acknowledgments.}
This work was partially supported by grant PAPIIT-UNAM IN 109415. 
CM was partially supported by the DFG SPP 2026 priority programme ``Geometry at infinity", and would like to kindly thank Anthony Phillips for multiple discussions and insights. JAZ was supported by a sabbatical grant by PASPA-UNAM.

\bibliographystyle{amsalpha}
\bibliography{Gauge}

\providecommand{\bysame}{\leavevmode\hbox to3em{\hrulefill}\thinspace}
\providecommand{\MR}{\relax\ifhmode\unskip\space\fi MR }
\providecommand{\MRhref}[2]{%
  \href{http://www.ams.org/mathscinet-getitem?mr=#1}{#2}
}
\providecommand{\href}[2]{#2}
\begin{thebibliography}{ORS05}

\bibitem[AL94]{AL94}
A.~Ashtekar and J.~Lewandowski, \emph{Representation theory of analytic
  holonomy {C}*-algebras}, Knots and quantum gravity, Clarendon press, Oxford
  (1994), 21--62.

\bibitem[Bae98]{Baez:1997zt}
J.~C. Baez, \emph{Spin foam models}, Class. Quant. Grav. \textbf{15} (1998),
  1827--1858.

\bibitem[Bar91]{Ba91}
J.~W. Barrett, \emph{Holonomy and path structures in general relativity and
  {Y}ang-{M}ills theory}, Int. J. Theor. Phys. \textbf{30} (1991), no.~9,
  1171--1215.

\bibitem[BH11]{BH11}
J.~C. Baez and J.~Huerta, \emph{An invitation to higher gauge theory}, Gen.
  Rel. Gravit. \textbf{43} (2011), no.~9, 2335--2392.

\bibitem[Bot56]{Bott56}
R.~Bott, \emph{An application of the {M}orse theory to the topology of {L}ie
  groups}, Bull. Soc. Math. France \textbf{84} (1956), 251--281.

\bibitem[CP94]{CP94}
A.~Caetano and R.~F. Picken, \emph{An axiomatic definition of holonomy}, Int.
  J. Math. \textbf{5} (1994), no.~6, 835--848.

\bibitem[Cre83]{Creu83}
M.~Creutz, \emph{Quarks, gluons and lattices (vol. 8)}, {C}ambridge
  {M}onographs on {M}ath. {P}hys., {C}ambridge {U}niversity {P}ress, 1983.

\bibitem[GP96]{GP96}
R.~Gambini and J.~Pullin, \emph{Loops, knots, gauge theories and quantum
  gravity}, {C}ambridge {M}onographs on {M}ath. {P}hys., {C}ambridge
  {U}niversity {P}ress, 1996.

\bibitem[KN63]{KN63}
S.~Kobayashi and K.~Nomizu, \emph{Foundations of differential geometry (vol. 1,
  no. 2)}, {N}ew {Y}ork, 1963.

\bibitem[Kob54]{Koba54}
S.~Kobayashi, \emph{La connexion des variet\'es fibr\'ees {I}, {II}}, Comptes
  Rendus Acad. Sci. \textbf{238} (1954), no.~3--4, 318--319,443--444.

\bibitem[Koz07]{Koz07}
D.~Kozlov, \emph{Combinatorial algebraic topology (vol. 21)}, {S}pringer
  {S}cience \& {B}usiness {M}edia, 2007.

\bibitem[L\"82]{Luscher82}
M.~L\"uscher, \emph{Topology of lattice gauge fields}, Commun. Math. Phys.
  \textbf{85} (1982), 39--48.

\bibitem[Las56]{Lash56}
R.~Laschof, \emph{Classification of fibre bundles by the loop space of the
  base}, Ann. of Math. \textbf{64} (1956), no.~3, 436--446.

\bibitem[Lef42]{Lef42}
S.~Lefschetz, \emph{Algebraic {T}opology (vol. 27)}, American Mathematical
  Society Colloquium Publications, v. 27., 1942.

\bibitem[Lew93]{Lew93}
J.~Lewandowski, \emph{Group of loops, holonomy maps, path bundle and path
  connection}, Class. Quantum Grav \textbf{10} (1993), no.~5, 879--904.

\bibitem[Mil56]{Milnor56}
J.~W. Milnor, \emph{Construction of universal bundles {I}-{II}}, Ann. of Math.
  \textbf{63} (1956), no.~2, 272--284,430--436.

\bibitem[Mil63]{Mil63}
\bysame, \emph{Morse theory}, {A}nnals of {M}athematics {S}tudies, No. 51,
  {P}rinceton {U}niversity {P}ress, 1963.

\bibitem[MS74]{MS74}
J.~W. Milnor and J.~D. Stasheff, \emph{Characteristic classes}, {A}nnals of
  {M}athematics {S}tudies, No. 76, {P}rinceton {U}niversity {P}ress, 1974.

\bibitem[NR61]{NR61}
M.~S. Narasimhan and S.~Ramanan, \emph{Existence of universal connections},
  Amer. J. Math. \textbf{83} (1961), no.~2, 563--572.

\bibitem[ORS05]{ORS05}
D.~Oriti, C.~Rovelli, and S.~Speziale, \emph{Spinfoam 2{D} quantum gravity and
  discrete bundles}, Class. Quantum Grav. \textbf{22} (2005), no.~1, 85--108.

\bibitem[Pac91]{Pach91}
U.~Pachner, \emph{{P}.{L}. homeomorphic manifolds are equivalent by elementary
  shellings}, European J. Combin. \textbf{12} (1991), no.~2, 129--145.

\bibitem[Phi85]{Phil85}
A.~Phillips, \emph{Characteristic numbers of ${U}_{1}$-valued lattice gauge
  fields}, Ann. Phys. \textbf{161} (1985), 399--422.

\bibitem[PS86]{PS86}
A.~Phillips and D.~Stone, \emph{{L}attice gauge fields, principal bundles and
  the calculation of topological charge}, Commun. Math. Phys. \textbf{103}
  (1986), no.~4, 599--636.

\bibitem[PS90]{PS90}
\bysame, \emph{{T}he computation of characteristic classes of lattice gauge
  fields}, Commun. Math. Phys. \textbf{131} (1990), no.~2, 255--282.

\bibitem[PS93]{PS93}
\bysame, \emph{{A} topological {C}hern-{W}eil theory}, Memoirs of the {AMS},
  no. 504, 1993.

\bibitem[Rei94]{Reisenberger:1994aw}
Michael~P. Reisenberger, \emph{{World sheet formulations of gauge theories and
  gra\-vity}}, {On recent developments in theoretical and experimental general
  relativity, gravitation, and relativistic field theories. Proceedings, 7th
  Marcel Grossmann Meeting, Stanford, USA, July 24-30, 1994. Pt. A + B}, 1994.

\bibitem[Seg68]{Segal68}
G.~Segal, \emph{Classifying spaces and spectral sequences}, Inst. Hautes
  \'Etudes Sci. Publ. Math. (1968), no.~34, 105--112.

\bibitem[Ste51]{Steen51}
N.~E. Steenrod, \emph{The topology of fibre bundles}, Vol. 14. Princeton
  University Press, 1951.

\bibitem[SW09]{SW09}
U.~Schreiber and K.~Waldorf, \emph{Parallel transport and functors}, J.
  Homotopy Relat. Struct. \textbf{4} (2009), no.~1, 187--244.

\bibitem[Zap12]{Zapata12}
J.~A. Zapata, \emph{Local gauge theory and coarse graining}, Journal of
  {P}hysics: {C}onference {S}eries \textbf{360} (2012), no.~1, 012054.

\end{thebibliography}

\end{document}